\newif\ifSubmission
 \Submissiontrue

\newif\ifComments
\Commentsfalse

\newif\ifAnonymous
\ifSubmission
    \Commentsfalse
\else
    \Commentstrue
    \Anonymousfalse
\fi

\documentclass[12pt]{article}
\usepackage{geometry}
\usepackage[utf8]{inputenc}

\ifAnonymous
    \author{}
\else
    \author{
    Andrea Coladangelo\thanks{University of Washington. Email: \texttt{coladan@cs.washington.edu}} 
    \quad Dakshita Khurana\thanks{University of Illinois Urbana-Champaign and NTT Research. Email: \texttt{dakshita@illinois.edu}} 
    \quad Saachi Mutreja\thanks{Columbia University. Email: \texttt{sm5540@columbia.edu}}\\
    \quad Bhaskar Roberts\thanks{UC Berkeley and NTT Research. Email: \texttt{bhaskarroberts@gmail.com}}
    \quad Joseph Slote\thanks{University of Washington. Email: \texttt{jslote@cs.washington.edu}}
    \quad Avishay Tal\thanks{UC Berkeley. Email: \texttt{atal@berkeley.edu}}}
\fi
\date{}

\usepackage{libertine}

\usepackage{amsthm}

\usepackage[T1]{fontenc}

\usepackage{enumitem}
\usepackage{graphicx} 
\usepackage{physics}
\usepackage{breakcites}
\usepackage{amsmath,thmtools,bm}
\usepackage{amssymb}
\usepackage[libertine]{newtxmath}
\usepackage{bm}
\usepackage[colorlinks=true, allcolors=blue]{hyperref}
\usepackage[dvipsnames]{xcolor}
\usepackage{bbm}
\usepackage{cleveref}
\usepackage{mathtools}
\usepackage{algorithm, algpseudocode, float, algorithmicx}
\usepackage{mdframed}
\usepackage{caption}

\usepackage{multirow}
\usepackage{comment}

\usepackage[english]{babel}

\newtheorem{theorem}{Theorem}[section]
\crefname{theorem}{theorem}{theorems}
\Crefname{theorem}{Theorem}{Theorems}

\newtheorem{question}{Question}

\newtheorem*{remark*}{Remark}

\newtheorem{construction}[theorem]{Construction}
\crefname{construction}{construction}{constructions}
\Crefname{construction}{Construction}{Constructions}

\newtheorem{definition}[theorem]{Definition}
\crefname{definition}{definition}{definitions}
\Crefname{definition}{Definition}{Definitions}

\newtheorem{lemma}[theorem]{Lemma}
\crefname{lemma}{lemma}{lemmas}
\Crefname{lemma}{Lemma}{Lemmas}

\newtheorem{corollary}[theorem]{Corollary}
\crefname{corollary}{corollary}{corollaries}
\Crefname{corollary}{Corollary}{Corollaries}

\newtheorem{claim}[theorem]{Claim}
\crefname{claim}{claim}{claims}
\Crefname{claim}{Claim}{Claims}

\newtheorem{specialclaim}{Claim}

\crefname{specialclaim}{claim}{claims}
\Crefname{specialclaim}{Claim}{Claims}
\newtheorem{speciallemma}[specialclaim]{Lemma}
\crefname{speciallemma}{lemma}{lemmas}
\Crefname{speciallemma}{Lemma}{Lemmas}

\newenvironment{proofofclaim}
  {\pushQED{\qed}%
   \begin{proof}[Proof of claim]}
  {\end{proof}}

\newcommand{\Prove}{\mathsf{Prove}}
\newcommand{\Verify}{\mathsf{Verify}}
\newcommand{\poly}{\mathsf{poly}}

\newcommand{\RO}{H}
\newcommand{\ro}{h}
\newcommand{\minent}{{h_\infty}}
\newcommand{\Minent}{{H_\infty}}

\newcommand{\secp}{\lambda}
\newcommand{\negl}{\mathsf{negl}}

\newcommand{\bit}{\{0,1\}}
\newcommand{\hw}{\operatorname{hw}}
\newcommand{\zero}{\mathbf 0}
\newcommand{\Decode}{\operatorname{Decode}}

\newcommand{\GRSUniqueDecode}{\operatorname{GRSUniqueDecode}}
\newcommand{\FRS}{\operatorname{FRS}}
\newcommand{\GRS}{\operatorname{GRS}}
\newcommand{\GoodErrors}{\mathsf{GoodErrors}}
\newcommand{\bbC}{\mathbb{C}}
\newcommand{\bbE}{\mathbb{E}}
\newcommand{\bbF}{\mathbb{F}}

\newcommand{\bbN}{\mathbb{N}}
\newcommand{\bbR}{\mathbb{R}}

\newcommand{\bfa}{\mathbf{a}}

\newcommand{\bfe}{\mathbf{e}}

\newcommand{\bfh}{\mathbf{h}}

\newcommand{\bfX}{\mathbf{X}}
\newcommand{\bfx}{\mathbf{x}}

\newcommand{\bfz}{\mathbf{z}}

\newcommand{\bfu}{\mathbf{u}}

\newcommand{\bfv}{\mathbf{v}}
\newcommand{\bfw}{\mathbf{w}}
\newcommand{\cE}{\mathcal E}
\newcommand{\Good}{\mathsf{Good}}
\newcommand{\Bad}{\mathsf{Bad}}

\newcommand{\BadErrors}{\mathsf{BadErrors}}
\newcommand{\cA}{\mathcal{A}}

\newcommand{\cD}{\mathcal{D}}

\newcommand{\cH}{\mathcal{H}}

\allowdisplaybreaks

\renewcommand{\L}{\mathsf{Light}}
\renewcommand{\H}{\mathsf{Heavy}}
\newcommand{\F}{\mathcal{F}}
\newcommand{\bias}{p}

\DeclareSymbolFont{yhlargesymbols}{OMX}{yhex}{m}{n} \DeclareMathAccent{\reallywidehat}{\mathord}{yhlargesymbols}{"62}

\ifComments
    \newcommand{\authnote}[3]{\textcolor{#3}{[{\footnotesize {\bf #1:} { {#2}}}]}}
\else
    \newcommand{\authnote}[3]{}
\fi

\title{Unconditional Certified Randomness without Structure}

\begin{document}

\setcounter{tocdepth}{2}
\maketitle

\begin{abstract}
    We obtain a certified randomness protocol in the quantum random oracle model. The protocol is non-interactive and publicly verifiable with a classical verifier, and is based on Yamakawa and Zhandry's proof of quantumness [JACM'24]. We prove unconditional security of this protocol against adversaries making subexponentially-many adaptive quantum queries to the random oracle.
    
    Prior work on certified randomness relative to a random oracle additionally assumed the Aaronson--Ambainis conjecture or proved security only against low query-depth adversaries.
\end{abstract}
\thispagestyle{empty}
\newpage
\tableofcontents
\thispagestyle{empty}

\newpage

\setcounter{page}{1}

\section{Introduction}\label[section]{sec:intro}
Randomness is an inherent feature of quantum mechanics: even complete knowledge of a quantum state does not, in general, determine the outcome of a measurement. Certified randomness takes this idea a step further: it allows a classical user to verify the randomness of a measurement produced by an untrusted quantum device. 

Protocols for certified randomness have a long tradition in quantum information theory. They were first studied in the so-called \emph{device-independent} setting, where a classical user interacts with two or more non-communicating quantum devices and certifies randomness through the violation of a Bell inequality~\cite{Col06, PAM+10, VV12, miller2016robust, miller2017universal, arnon2018practical}. More recently, a line of work initiated by Brakerski et al.~\cite{BCM+21} has shown that randomness can also be certified using a \emph{single} quantum device, replacing the physical assumption of non-communication by an assumption on the computational power of the device. The first protocols were based on the quantum hardness of the Learning With Errors problem (LWE) \cite{mahadev2022efficient}. In a later work, Aaronson and Hung~\cite{AH23} showed that sampling-based quantum advantage experiments can also be turned into certified randomness protocols, under a novel, non-standard complexity-theoretic assumption.

Both the device-independent and computational approaches have limitations and leave important desiderata unmet. 
In the device-independent setting, while security is information-theoretic, it relies on physically enforcing a non-communication condition between spatially separated devices, which is generally impractical. The existing LWE-based protocols~\cite{BCM+21, mahadev2022efficient} require some rounds of interaction with the quantum device, in between which the device should maintain a coherent state. Moreover, they produce randomness that is only privately verifiable (using a trapdoor held only by the verifier). The sampling-based protocol of \cite{AH23} has the advantage of being likely within reach of the capabilities of current devices, but verification requires exponential time. 

\paragraph{Certifiable randomness in the quantum random oracle model (QROM).} Recently, Yamakawa and Zhandry~\cite{YZ24} (henceforth YZ, for short) proposed an alternative approach to certified randomness. This came as a (conditional) corollary of their breakthrough complexity theory result, which showed an unconditional quantum advantage for an NP-search problem in the quantum random oracle model (QROM). Informally, the NP-search problem is the following: given a suitable code $C \subset \Sigma^n$, find a codeword $x \in C$ such that $H_i(x_i) = 0$ for all coordinates $i \in [n]$, where $H = (H_1, \ldots, H_n):\Sigma^n\to\{0,1\}^n$ is a random oracle mapping alphabet symbols to bits. They then showed that there exists an efficient quantum algorithm solving the problem with polynomially many queries to $H$, whereas exponentially many queries are necessary for any classical algorithm. 
 
Now, the YZ quantum algorithm has the following curious property: it does not return a fixed correct solution, but rather an (essentially) uniformly random one. As pointed out by Yamakawa and Zhandry, the Aaronson--Ambainis conjecture~\cite{AA14} implies that this is not a coincidence: for \emph{any} quantum advantage relative to a random oracle alone, the quantum algorithm that solves the problem must generate some entropy in its output. Applied to the YZ problem, this yields a certified randomness protocol assuming the Aaronson--Ambainis conjecture is true. 

Moreover, this protocol has very desirable features. It involves a \emph{single prover} (so no communication assumption is needed). It is \emph{non-interactive}: the quantum device simply produces a candidate solution, with no further interaction with the verifier.
It is \emph{publicly verifiable}, since anyone with access to the random oracle can efficiently check whether the reported output is a valid solution. So reliance on the Aaronson--Ambainis conjecture appears to be the main shortcoming of the YZ certified randomness protocol. While this conjecture is widely believed to be true, it has resisted proof despite concerted study; \textit{e.g.}~\cite{montanaro2012some,BansalSinhaDeWolf,LovettZhang,GutierrezAA, bhattacharya2025random}. Motivated by the above discussion, we focus here on the following question:
\begin{center}
\emph{Does certified randomness exist unconditionally in the QROM?\\} 
\end{center}
In particular, we would like such a protocol to retain the desirable features mentioned above: single-prover, non-interactivity, and public verifiability. An affirmative answer to this question would also provide further evidence for the veracity of the Aaronson--Ambainis conjecture, as it would settle one of its important consequences.

\subsection{Our result}
In this work, we answer the above question affirmatively. We show that a slight variant of the YZ problem unconditionally certifies randomness in the QROM.
Moreover, the resulting protocol is non-interactive and publicly verifiable.

The variant differs from the original YZ problem only in the choice of code $C$ and the distribution of the random oracle $H = (H_1, \ldots, H_n)$ (which in turn can be simulated in the standard QROM without loss).
In the original problem, each $H_i : \Sigma \rightarrow \{0,1\}$ maps every symbol independently to a uniformly random bit.
In our variant, the bits are instead \emph{biased}: independently for every input, $H_i$ outputs $0$ with probability $1-p$, and $1$ with probability $p$, where $p = 1/\Theta(\sqrt{n})$.
We note that this biased random oracle can be constructed straightforwardly in the QROM itself.
We discuss the new choice of code $C$ and its properties in the technical overview (Section~\ref{sec:overview}). We refer to this YZ variant as the \emph{biased YZ problem}. We show the following:
\begin{theorem}[Informal]
\label{thm:main-informal}
The biased YZ problem in the QROM is a certified min-entropy protocol (as formally defined in Definition~\ref{def:certifiable-min-entropy}) with the following guarantee: for every $c\in(0,1/2)$, the protocol yields min-entropy at least $\Omega(n^c)$ with security against adversaries making at most $2^{o(n^c)}$ adaptive queries.
The protocol is single-prover, non-interactive, and publicly verifiable.

\end{theorem}
We elaborate further on the notion of a certified min-entropy protocol (introduced in \cite{YZ24}). Informally, a certified min-entropy protocol in the QROM provides the following information-theoretic guarantee: subject only to a bound on the number of quantum queries made by the prover, with high probability over sampling a random oracle $H$ and conditioned on acceptance in the protocol, the verifier's output must have high min-entropy. Our certified min-entropy protocol has a particularly simple form: the prover returns a solution to the biased YZ problem (which is defined with respect to $H$), and, upon checking that it is valid, the verifier outputs the solution itself.

Therefore, put more thoroughly, the guarantee of \Cref{thm:main-informal} is as follows: for any $c\in(0,1/2)$ and any (even computationally unbounded) adversary $\mathcal{A}$ making \(2^{o(n^c)}\) quantum queries to the random oracle, with overwhelming probability over $H$, the following holds: if $\mathcal{A}^H$ is accepted with any noticeable probability, the min-entropy of its output, conditioned on acceptance, must be at least $\Omega(n^c)$.

Our starting point is a recent result of Khurana, Roberts, and Tal~\cite{KRT},
showing that the original YZ problem certifies min-entropy against provers who make their $\poly(n)$ queries in at most $o(\log n)$ adaptive layers. 
In the present work, we are able to lift this restriction on adaptivity entirely, albeit for a slight variant of the YZ problem. As a result, we obtain a protocol to certify min-entropy (in the QROM) even against provers that make $2^{o(n^c)}$ \textit{adaptive} queries (a doubly exponential improvement).

Our proof is cast in the formalism of hybrid arguments and oracle reprogramming.
The main innovation is a relaxed reprogramming guarantee that holds even when modifying the oracle at heavily-queried coordinates: we show that any string $\bfx\in\Sigma^n$ with heavy query weight on many of its symbols retains significant query weight under many different reprogrammings.
From a technical standpoint, a key perspective shift is to work with the square root of query weight, which we term the ``query norm.''
We then reason about the ``query norm polytope,'' defined by a system of linear constraints that limit how query norms can change between pairs of oracles.
We provide an overview of these proof ideas in the next section.

Finally, we remark that the YZ quantum advantage is remarkable for another reason: it remains the only known example of a genuinely ``structure-less'' quantum advantage, in the sense that it arises relative to a random oracle without leveraging any additional algebraic or computational structure. Since its introduction, the YZ construction has spurred a broader effort to understand this new source of quantum advantage. Most recently, it has inspired ``Decoded Quantum Interferometry'' (DQI)~\cite{jordan2025optimization}, a new framework for quantum algorithms for optimization (which recovers the YZ algorithm as a special case). Thus, beyond providing a certified randomness protocol with several desirable properties, our result makes progress in understanding a key structural property---inherent randomness---of this new source of quantum advantage (albeit for a slight variant of the original YZ problem).

\section{Overview}
\label{sec:overview}

\subsection{The Yamakawa--Zhandry search problem}
Let us recall the general structure of the YZ search problem~\cite{YZ24}.
Let $C\subseteq \Sigma^n$
be a code over a large alphabet $\Sigma$, and let
\[
H=(H_1,\ldots,H_n),\qquad H_i:\Sigma\rightarrow\{0,1\},
\]
be a random oracle.
We say a string $\bfx=(x_1,\ldots,x_n)\in \Sigma^n$ is \emph{correct} if
\[
\bfx\in C
\qquad\text{and}\qquad
H_i(x_i)=0\quad\text{for every }i\in[n].
\]
The \emph{YZ search problem} is to find any correct $\bfx$, given quantum query access to $H$.

For an appropriate choice of code $C$, Yamakawa--Zhandry gave an algorithm (henceforth the \emph{YZ algorithm}) which finds a correct $\bfx$ using a single layer of parallel quantum queries.
An important property of $C$, which will become crucial later on, is that it is \textit{list-recoverable}: given short lists $S_i$ of possible symbols at each coordinate $i\in[n]$, there are very few codewords that are compatible with most of the $S_i$'s; \textit{i.e.}, there are few $\bfx =(x_1,\ldots, x_n)\in C$ for which $x_i\in S_i$ for many $i$.

The honest prover prepares the uniform superposition over \(C\) and, independently for each coordinate, the uniform superposition over symbols hashing to zero. After Fourier transforming both registers, the first register is supported on \(C^\perp\), while the second behaves like a sparse error vector because of the oracle bias. The prover uses error-correcting properties of the dual code \(C^\perp\) to perform convolution, effectively computing the unitary
\[\sum_{c \in C^{\perp}} \alpha_c \ket{c} \otimes \sum_{e} \beta_e \ket{e} \mapsto \sum_{c \in C^{\perp}, e} \alpha_c \beta_e \ket{c+e} \]
By the convolution theorem, performing a Fourier transform of the above state yields the {\em intersection} of the primal values on both registers. Thus, measuring the resulting state outputs \(\bfx \in C\) such that $\forall i \in [n], H_i (x_i) = 0$.

\subsection{The min-entropy proof strategy at a high level}
For each oracle $h$, the YZ algorithm returns a sample from a distribution that is essentially uniform over the set of correct $\bfx$'s.
In recent work, Khurana, Roberts, and Tal~\cite{KRT} proved that this randomness is inherent in the output of any algorithm from a certain restricted class that correctly solves the YZ problem.
More specifically, \cite{KRT} proved unconditionally that the output of any successful, query-efficient prover has high min-entropy, provided the prover's $\poly(n)$-many queries are made in at most $o(\log n)$ adaptive layers.
In a nutshell, the main difficulty with many layers of adaptive queries is that the prover's query weight may depend on the oracle itself. Reprogramming the oracle may therefore change not only the prover's final output, but also the locations queried in all subsequent layers. The proof in~\cite{KRT} handled these dependencies via a layer-by-layer bootstrapping argument, which limited them to $o(\log n)$ adaptive layers. 
Our argument shares the same high-level template as \cite{KRT}, but develops a new reprogramming argument in order to obtain the full query-bound.

As is typical with reprogramming arguments, we begin by defining a  ``currency'' that the prover must pay to query any oracle entry $H_i(x)$, and whose total amount is controlled by the total number of queries $Q$ made by the prover.
We then derive a contradiction by analyzing how this currency would have to be distributed across the coordinates $i$ and symbols $x_i$ from any correct, high-probability codeword $\bfx$.

In the present work, this currency is the \emph{query norm} $\widetilde{w}_{i,x}^j$, which is charged at each query $j\in[Q]$ according to overlap of the pre-query state $\ket{\psi_{j-1}}$ with each of the oracle indices $(i,x)\in [n]\times \Sigma$:
\[\widetilde{w}_{i,x}^{\,j}:=\left\|\Pi_{i,x}\cdot\ket{\psi_{j-1}}\right\|_2\,.\]
Here $\Pi_{i,x}:=\ketbra{i,x}\otimes \mathbb{I}$, where the projection $\ketbra{i,x}$ acts on the oracle register and $\mathbb{I}$ acts on the algorithm's work register.
The query norm appears as an intermediate quantity in many hybrid arguments and is the square root of the familiar \emph{query weight}, $w_{i,x}^j:=(\widetilde{w}_{i,x}^{\,j})^2$.
It is elementary to see that the query norm across all queries and oracle coordinates is controlled by the number of queries as follows:
\[\textstyle\sum_{j,i,x}(\widetilde{w}_{i,x}^{\,j})^2\leq Q\,.\]

To discuss how query norm is distributed across oracle indices, it is useful to denote by \emph{total query norm} the cumulative quantity $\widetilde{w}_{i,x}:=\sum_{j=1}^Q\widetilde{w}_{i,x}^{\,j}$ paid to each oracle coordinate $(i,x)$.
This allows us to designate certain strings $\bfx\in \Sigma^n$ as being ``heavily queried'' in the following sense.
For any string $\bfx = (x_1,\ldots, x_n) $, we say $\bfx$ is \emph{$(m, t)$-heavily queried} by the algorithm $\mathcal{A}^h$ if $m$ coordinates of $\bfx$ each receive total query norm at least $t$ under $h$:
\[\bfx \text{ is }(m,t)\text{-heavily queried}\quad\iff\quad\# \big\{i\in[n]: \widetilde{w}_{i,x_i}\geq t\big\}\geq m\,.\]
We will mostly be interested in $\big(0.9n,\beta\big)$-heavily queried $\bfx$'s for a fixed, noticeable $\beta\in 1/\poly(n,Q)$ to be chosen later.
For this canonical parameter choice we will just say $\bfx$ is \emph{heavily queried}.

In the vocabulary just established, our argument splits into two main claims:
\begin{quote}
\vspace{-4\parskip}
\begin{specialclaim}[Low-entropy provers heavily query correct answers]
\label{clm:claim1-intro}
    With high probability over the random oracle $h\gets H$, if $\mathcal{A}^h$ returns a correct $\bfx$ with noticeable probability (\textit{i.e.}, $\mathcal{A}$ has low min-entropy under $h$), then $\bfx$ must be heavily queried.
\end{specialclaim}

\begin{specialclaim}[Query-bounded provers cannot heavily query a correct answer]
\label{clm:claim2-intro}
Suppose $\mathcal{A}$ makes no more than $2^{o(n^c)}$ queries.
Then with high probability over $h\gets H$, no correct $\bfx$ is heavily queried by $\mathcal{A}^h$.
\end{specialclaim}
\end{quote}
\noindent Instantiated appropriately, these statements together yield the quoted lower bound in \Cref{thm:main-informal}.

We remark that an important difference between the approach used here and the one used in~\cite{KRT} is that in the two claims above, as well as their proof, we stick with query \emph{norm}, rather than passing to query \emph{weight} $w_{i,x}^{\,j}:=(\widetilde{w}_{i,x}^{\,j})^2$, as is typical with hybrid arguments.
While relaxing hybrid argument bounds to query weight is often sufficient (and indeed is sometimes tight---\textit{e.g.}, as in the original \cite{BBBV97} bound for unstructured search), an analysis in terms of query weight seems to be too loose to obtain effective reprogramming bounds in our setting.

\Cref{clm:claim1-intro,clm:claim2-intro} are both proved via reprogramming arguments, as we outline next.
We will need the following elementary \emph{swapping lemma}, which is standard in hybrid arguments (see \textit{e.g.} \cite[Lemma 2]{Vaz98}, \cite[Theorem 3.3]{BBBV97}).
If oracles $h$ and $h'$ differ on a set
$X\subseteq[n]\times\Sigma$, then the Euclidean distance between the
adversary's states after \(j\) queries satisfies
\[
\left\|\,
\ket{\psi_{\smash{j}}^{\smash{h}}}-\ket{\psi_{\smash{j}}^{\smash{h'}}}\mkern1.5mu
\right\|_2
\leq
2\sum_{j'=1}^{j}
\sum_{(i,x)\in X}
\widetilde{w}_{i,x}^{\,j'}(h) \leq 2\sum_{(i,x)\in X}\widetilde{w}_{i,x}(h).
\]
Since this Euclidean distance upper-bounds the trace distance, it also
upper-bounds the TV distance of any measurement's outputs for the two states.

We first explain how to carry out the proofs of \Cref{clm:claim1-intro,clm:claim2-intro}  over the \emph{unbiased oracle distribution} that assigns each oracle bit independently and uniformly at random.
These analyses will eventually require slight modifications to apply to a biased random oracle, as we explain in \Cref{sec:overview-biased-security}.

\subsubsection{Proving \Cref{clm:claim1-intro}: Low-entropy provers heavily query correct answers}
\label{sec:claim1-intro}
\paragraph{Setup and bad\textsubscript{1} oracles.}
Fix a noticeable threshold $\alpha:=4n\beta$ and call an oracle $h$ ``bad\textsubscript{1}'' if there exists an $\bfx$ such that $\bfx$ is correct for $h$, $\Pr[\mathcal{A}^h=\bfx]\geq \alpha$, yet $\bfx$ is \emph{not} heavily queried by $\mathcal{A}^h$. (We consider this case bad since the output $\mathcal{A}^h$ does not have high min-entropy.)
Our goal is to establish an upper bound on $\Pr[H\text{ is bad\textsubscript{1}}]$.

For each bad\textsubscript{1} $h$, fix an $\bfx_h$ that is correct for $h$, returned with probability at least $\alpha$ by $\mathcal{A}^h$, and not heavily queried.
We will argue there is a large set of oracles $\mathcal{F}(h)$ such that $\bfx_h$ is still returned with noticeable probability: $\Pr[\mathcal{A}^{g}=\bfx_h]\geq \alpha/2$ for all $g\in\mathcal{F}(h)$.
Moreover, we will argue that for bad\textsubscript{1} oracles $h$ and $h'$, all the pairs $(g,\bfx_h)$, $g\in \mathcal{F}(h)$ and $(g',\bfx_{h'})$,  $g'\in\mathcal{F}(h')$ are \emph{distinct}.

Supposing we can establish a lower bound $|\mathcal{F}(h)|\geq K$ for all bad\textsubscript{1} $h$, we thus conclude:
\begin{align*}
    1&\geq \sum_{h \text{ bad}}\sum_{g\in \mathcal{F}(h)} \Pr[H=g\text{ and }\mathcal{A}^H=\bfx_h] \tag{Disjointness}\\
    &= \sum_{h \text{ bad}}\sum_{g\in \mathcal{F}(h)} \Pr[H=g]\Pr[\mathcal{A}^g=\bfx_h]\\
    &\geq \sum_{h\text{ bad}}\sum_{g\in \mathcal{F}(h)}\frac{\alpha}{2}\Pr[H=h]\tag{$H$ is unbiased, so $\Pr[H=h] = \Pr[H=g]$}\\
    &\geq \frac{K\alpha}{2}\Pr[H \text{ is bad\textsubscript{1}}],
\end{align*}
from which we obtain the desired upper bound on the fraction of bad oracles:
\[\Pr[H \text{ is bad\textsubscript{1}}] \leq \frac{2}{K\alpha}\,.\]
For the right value of $K$, this gives the desired bound on the fraction of bad oracles. We will argue that $K$ will be large enough so that this value is small.

\paragraph{Reprogramming.}
The large family $\mathcal{F}(h)$ is derived as follows.
Because $\bfx_h$ is not heavily queried, there is a set of coordinates $I\subseteq[n]$, $|I|=0.1n$, such that each coordinate-symbol of $\bfx_h$ in $I$ is \textit{lightly queried}:
\[
\widetilde{w}_{i,x_i}(h)
<
\beta=\frac{\alpha}{4n}
\qquad\text{for every }i\in I.
\]
We use $I$ to produce a large number of oracles on which the adversary still outputs $\bfx_h$ with high probability, while making sure to retain the disjointness property.
This is done by reprogramming: for every subset of low total query norm coordinates, 
$S\subseteq\{(i,x_i):i\in I\}$,
let $h^S$ be the oracle obtained by changing the oracle entries at all $(i,x)\in S$ from zero to one.
The swapping lemma gives
\[
\begin{aligned}
\left\|\,
\ket{\psi_Q^h}-\ket{\psi_Q^{h^S}}
\,\right\|_2
&\overset{\text{(Swapping Lemma)}}{\leq}\;
2\sum_{(i,x_i)\in S}
\widetilde{w}_{i,x_i}(h)\;\;<\;\;
2\cdot 0.1n\cdot\frac{\alpha}{4n}
\;\;\leq\;\;
\frac{\alpha}{2}.
\end{aligned}
\]
Consequently, the adversary continues to output $\bfx_h$ with probability
at least $\alpha/2$ under every oracle $h^S$, and with $\mathcal{F}(h):=\{h^S\}_{S}$, we have $|\mathcal{F}(h)|\geq 2^{0.1n}=:K$.
Finally, distinctness of $(g,\bfx)$ pairs follows from the fact that for any $(g,\bfx)$ in any family, the original oracle $h$ for which $g\in\mathcal{F}(h)$ may be uniquely recovered by resetting the values of $g$ to zero at the indices $(i,x_i)$ from $\bfx = (x_1,\ldots, x_n)$.

We thus conclude that except for a negligible fraction of all oracles $h$, all correct $\bfx$ output by $\mathcal{A}^h$ with probability at least $\alpha$ must be heavily queried.



\subsubsection{Proving \Cref{clm:claim2-intro}: Query-bounded provers cannot heavily query a correct answer}
\label{sec:claim2-intro}
We now provide an overview of the proof of \Cref{clm:claim2-intro}, which will be the most substantial part of our argument.
\paragraph{Setup and bad\textsubscript{2} oracles.} Our goal is to prove that for almost all oracles there is no heavily queried correct $\bfx$.
At a high level, this is achieved by the following counting argument.
Call an oracle $h$ ``bad\textsubscript{2}'' if $\mathcal{A}^h$ does manage to heavily query a correct $\bfx$ for $h$, and set up a bipartite graph with these bad\textsubscript{2} oracles as the left vertices and all oracles as the right vertices.
We will define edges via reprogramming, and argue that the degrees of the left vertices are uniformly large, say at least $D_\mathsf{L}$, while the degrees of the right vertices are uniformly small, say at most $d_\mathsf{R}$.
Counting the number of edges in two ways---based on the left- and right-degrees respectively---we find
\[D_\mathsf{L}\cdot|\text{bad\textsubscript{2}}|\leq \text{\# of edges}\leq d_\mathsf{R}\cdot|\text{all}|,\]
from which we conclude that the fraction of bad\textsubscript{2} oracles is at most $d_\mathsf{R}/D_\mathsf{L}$.

\paragraph{Lower-bounding the left-degree.}
Let $h$ be a bad\textsubscript{2} oracle.
We would like to associate $h$ with many alternative oracles $g$ while ensuring that each $g$ isn't associated with too many bad\textsubscript{2} $h$'s.
However, unlike in the first claim, we are not guaranteed a supply of ``lightly queried'' coordinates in $\bfx$---in fact there may be none at all.
Thus we cannot use the swapping lemma, which essentially relies on $\mathcal{A}$ ``not noticing'' the reprogramming, and no such guarantee is available for heavy coordinates.
This is especially true in adaptive algorithms: because the adversary is making $Q$ adaptive queries, \textit{a priori} there is nothing preventing $\mathcal{A}$ from using information gained from a heavy query on one entry of $h$ to detect that $h$ has been modified, and then changing its behavior completely.

Despite these challenges, we are able to obtain the following key result:
\begin{quote}
    \vspace{-4\parskip}
    \begin{speciallemma}[Heavy-coordinate reprogramming, informal version of \Cref{lemma:subcube with square root heavy coordinates}]
    \label{lem:heavy-repro}
    For any oracle $h$ and any $\bfx$ that is $(0.9n,\beta)$-heavily queried, there are $2^{\Omega(n)}$-many oracles $g$, each obtained by reprogramming $h$ at a subset of the coordinate-symbol pairs $(i,x_i)$ of $\bfx$, such that
    \[\bfx \text{ is }\left(\Omega(\sqrt{n}),\Omega\!\left(\tfrac{\beta}{Q^2n^2}\right)\right)\!\text{-heavily queried under $g$.}\]
    \end{speciallemma}
\end{quote}

The heavy-coordinate reprogramming lemma defines the edges in the bipartite graph under discussion and thus directly yields a $2^{\Omega(n)}$ lower bound on $D_\mathsf{L}$.
One notices that in $g$, the number of heavy coordinates is $\Omega(\sqrt{n})$, which is quadratically worse than in the original $h$.
Looking ahead, this will cause some difficulty with bounding $d_\mathsf{R}$.

This lemma is the technical heart of our argument, and we outline its proof in the next subsection.

\paragraph{Upper-bounding the right-degree.}
Note that while $\bfx_h$ is generally no longer correct for the oracles $g$ returned by the heavy-coordinate reprogramming lemma, it is of course still in the code $C$.
Therefore to bound the right-degree, it suffices to argue:
\begin{quote}
	\textbf{Claim.} For any oracle $g$, there are very few $\bfx$’s that (\textit{i}) are $\big(\Omega(\sqrt{n}),1/\poly(n,Q)\big)$-heavily queried by $\mathcal{A}^g$ and (\textit{ii}) are in the code $C$.
\end{quote}
This is proved as follows.
By an elementary calculation, the fact that $\mathcal{A}^g$ makes at most $Q$ queries means there are few oracle coordinates $(i,x)$ with noticeable total query norm.\footnote{Indeed, for every query $j$,
we have $\sum_{(i,x)\in[n]\times\Sigma}
(\widetilde{w}^{\,j}_{i,x})^2
=
1$, and so for every $\eta>0$, there are at most $Q^2/\eta^2$
oracle inputs whose total query norm is at least $\eta$.}
Now form the lists  $\{S_i\}_{i=1}^n$, where $S_i$ contains those $x\in \Sigma$ such that $(i,x)$ has noticeable total query norm.
By definition, any $\big(\Omega(\sqrt{n}),1/\poly(n,Q)\big)$-heavily queried $\bfx=(x_1,\ldots, x_n)$ has $\Omega(\sqrt{n})$ coordinates such that $x_i\in S_i$.
It remains only to invoke the list-recovery property of our code $C$, which precisely stipulates that there are few $\bfx$ in $C$.

\bigskip

At this point a reader familiar with the details of the code $C$ employed in the original YZ problem may notice an issue: the code in the original YZ problem has inadequate list recovery parameters to be applied in the argument above; in particular, it can only limit the number of codewords that are consistent with a constant fraction of lists $S_i$.
On the other hand, the parameters of the heavy coordinate reprogramming lemma require us to control the number of codewords consistent with a vanishing $O(1/\sqrt{n})$ fraction of $S_i$'s,
and it is not clear how to strengthen the lemma to loosen our requirements.

Our solution will be to circumvent the difficulty by modifying $C$ to improve its list recovery parameters, which allows the above argument to go through.
However, the modification also creates a new issue: the \emph{completeness} of the protocol is no longer guaranteed, because the honest YZ algorithm no longer succeeds with high probability.
Our solution to this is to replace the uniformly random oracle with a \textit{biased} random oracle, which recovers the YZ algorithm's high probability of success.
A consequence of these complications is that the counting arguments just described, both for \Cref{clm:claim1-intro} and for \Cref{clm:claim2-intro}, must become \textit{weighted} counting arguments to match the biasedness of the oracles---but this final adjustment is fairly straightforward.

We now turn to explaining the ideas behind the heavy-coordinate reprogramming lemma.
Then we detail the necessary code modifications and oracle biasing.

\begin{remark*}The use of biased oracles and codes with better list-recovery seems to be an emerging theme for the YZ problem.
Related modifications to the YZ problem have appeared in recent works, both in communication complexity \cite{goosQTFNP} and for separating QCMA from QMA \cite{BHV26}.
\end{remark*}

\subsection{Reprogramming heavy coordinates and the query-norm polytope}
Here we explain our approach to reprogramming oracles on coordinates receiving large total query norm (\Cref{lem:heavy-repro}), beginning with a new geometric perspective afforded by the query norm.

\paragraph{Local stability of query norms.}
As remarked above, the swapping lemma is inadequate for understanding the distribution of total query norm on coordinates of a given $\bfx$ after reprogramming the oracle at locations with large query norm.
We obtain additional control by applying the swapping lemma to the intermediate states of the computation as follows.
Let $h$ and $h'$ differ
on a set $X\subseteq[n]\times\Sigma$.
For any oracle input $u=(i,x)$, the reverse triangle inequality gives
\[
\begin{aligned}
\left|
\widetilde{w}^{\,j}_{u}(h)
-
\widetilde{w}^{\,j}_{u}(h')
\right|
&=
\left|\,
\left\|\Pi_u\ket{\psi^{\smash{h}}_{\smash{j-1}}}\right\|_2
-
\left\|\Pi_u\ket{\psi^{\smash{h'}}_{\smash{j-1}}}\right\|_2\,
\right|\\
&\leq
\left\|\,
\ket{\psi^{\smash{h}}_{\smash{j-1}}}
-
\ket{\psi^{\smash{h'}}_{\smash{j-1}}}
\,\right\|_2\\
&\leq
2\sum_{t<j}\sum_{v\in X}
\widetilde{w}^{\,t}_{v}(h),
\end{aligned}
\]
where the last inequality is the swapping lemma applied immediately
before the \(j\)-th query.  Summing over \(j\in[Q]\) and applying another triangle inequality yields
\[
\left|
\widetilde{w}_{u}(h)
-
\widetilde{w}_{u}(h')
\right|
\leq
2Q\sum_{v\in X}\widetilde{w}_{v}(h).
\]
We call this inequality the \emph{local stability of query norms}.  It says that
if the inputs on which two oracles differ have small query norm under
$h$, then the total query norm assigned to any fixed input, $u$, cannot
change substantially when passing from $h$ to $h'$.

\paragraph{The query-norm polytope.}
Simultaneously applying local stability across all oracles yields a set of linear constraints whose feasible region we refer to as the \emph{query-norm polytope}.
With $\mathcal{H}$ denoting the set of all oracles and writing $\widetilde{w}_{i,x}(h)$ to index the $(i,x,h)$\textsuperscript{th} entry of a vector $\widetilde{w}\in\mathbb{R}^{n\cdot|\Sigma|\cdot|\mathcal{H}|}$, the query-norm polytope is given by:
\[\mathbf{P}_{Q}= \left\{\;\widetilde{w}\in \mathbb{R}^{n\cdot|\Sigma|\cdot|\mathcal{H}|}\;\;\middle|\;\;
\begin{aligned}&\widetilde{w}_{i,x}(h) \in [0,Q]\quad\text{for all}\quad i\in[n],x\in\Sigma,h\in\mathcal{H},\\[0.5em]
&|\widetilde{w}_{i,x}(h)-\widetilde{w}_{i,x}(h^S)|\leq 2Q\textstyle\sum_{(i',x')\in S}\widetilde{w}_{i',x'}(h)\\
&\text{\hspace{7.3em}for all}\quad i\in [n],x\in\Sigma,h\in\mathcal{H}, S\subseteq [n]\times \Sigma
\end{aligned}\;\right\}\,.\]
As before, $h^{S}$ denotes the oracle obtained from $h$ by flipping its output on every input in $S$.

We will be interested in the constraints imposed by the query-norm polytope on a family of reprogrammings.
To that end, fix an oracle $h$ and consider a set $M\subseteq[n]\times\Sigma$
such that
\[\widetilde{w}_{u}(h)\geq\beta \quad\text{for every}\quad u\in M\,.\]

The main consequence of the query-norm polytope is that query norm cannot
disappear from all of \(M\) under a large family of reprogrammings.
More precisely, let \(k\) be any integer satisfying
\[
1+\binom{k}{2}\leq |M|.
\]
The main claim internal to the proof of the heavy-coordinate reprogramming lemma is that there exists a set of ``protected'' inputs
\[
A_k\subseteq M,
\qquad
|A_k|\leq \binom{k}{2},
\]
such that the following holds.  For every subset
$S\subseteq M\setminus A_k$,
the adversary interacting with the reprogrammed oracle \(h^S\) assigns
query norm at least
\[
\tau
:=
\frac{\beta}{(8Q|M|)^2}
\]
to at least \(k\) inputs in \(M\).  In other words, after protecting
only \(\binom{k}{2}\) inputs, we may freely reprogram an arbitrary
subset of all remaining inputs in $M$ while guaranteeing that at least \(k\)
inputs in \(M\) continue to have noticeable query norm after reprogramming
(though the identities of these $k$ still-noticeable coordinates may differ across reprogrammed oracles $h^S$).
Next, we provide some intuition for why this is true.

\paragraph{Intuition for the polytope argument.}
Consider first the case \(k=1\).  Suppose that some reprogramming made
every input in \(M\) have very small query norm.  Starting from this
reprogrammed oracle \(h'\), reverse the reprogramming to go back to \(h\).
Because all the inputs being changed are light under the reprogrammed
oracle, local stability says that this reversal cannot substantially
increase the query norm at any particular input.  But this contradicts the
assumption that every input in \(M\) has query norm at least \(\beta\)
under \(h\).  Thus, every reprogramming leaves at least one input in
\(M\) noticeably queried.

The general statement iterates this observation.  If every
reprogramming already leaves \(k\) heavy inputs, there is nothing to
prove.  Otherwise, fix an oracle \(h^\star\) having fewer than \(k\)
heavy inputs and partition \(M\) into the sets
\[
\mathsf{Heavy}
:=
\left\{
u\in M:
\widetilde{w}_u(h^\star)\geq\tau
\right\}
\qquad\text{and}\qquad
\mathsf{Light}
:=
M\setminus\mathsf{Heavy}.
\]

\(\mathsf{Heavy}\) contains at least \(1\) input and at most
\(k-1\) inputs.  
A two-step hybrid argument based on local stability shows that any reprogramming of $H$ on arbitrary set of inputs in \(\mathsf{Light}\) cannot make every input in \(\mathsf{Heavy}\) light, i.e., 
at least one of the heavy inputs of \(h^\star\) must survive. 
We then
apply the argument recursively to \(\mathsf{Light}\), obtaining another
\(k-1\) heavy inputs there at the cost of protecting $\binom{k-1}{2}$ inputs.

To make the conclusion uniform over all reprogrammings, we protect the
at most \(k-1\) inputs in \(\mathsf{Heavy}\), together with the inputs
protected by the recursive application.  Thus, the number of protected
inputs satisfies
\[
|A_k|
\leq
(k-1)+\binom{k-1}{2}
=
\binom{k}{2}.
\]
Note here that guaranteeing \(k\) surviving heavy inputs costs
\(O(k^2)\) protected inputs.

\paragraph{The remaining bottleneck.}
Suppose that \(M\) contains a linear number of symbols of a codeword.
To leave a linear number of symbols freely reprogrammable, the
polytope argument above allows us to set
\[
k=\Theta(\sqrt n).
\]
Consequently, we obtain exponentially many reprogrammed oracles, under every one
of which the adversary continues to place noticeable query norm on
\(\Theta(\sqrt n)\) symbols of the same codeword.
This is precisely what was promised by the heavy-coordinate reprogramming lemma, \Cref{lem:heavy-repro}.

Unfortunately, as mentioned above, this falls short of what is needed to use the code from the prior work.
The original Yamakawa--Zhandry code is list recoverable when a codeword agrees with the supplied symbol lists on a linear number of coordinates.
In contrast, our polytope argument only preserves \(\Theta(\sqrt n)\) heavily queried symbols.
One could hope to prove a stronger reprogramming statement that preserves a linear number of heavy inputs while still leaving a linear number of inputs freely
reprogrammable.
We record this as an open question:
\begin{question}
    \label{q:heavy-reprogramming}
    Can the heavy-coordinate reprogramming lemma (\Cref{lem:heavy-repro}) be improved to the following?
    Let $h$ be any oracle and let $\bfx \in \Sigma^n$ that is $(\Omega(n),\beta)$-heavily queried by a $Q$-query algorithm $\mathcal{A}^h$.
    Then there are $2^{\Omega(n)}$-many oracles $g$ for which $\bfx$ is $\big(\Omega(n),\Omega({\beta}/\poly(Q,n))\big)$-heavily queried.
\end{question}
\noindent Note that it is possible that \Cref{q:heavy-reprogramming} could be resolved using nothing more than the query-norm polytope framework established above, which is convenient because it reduces \Cref{q:heavy-reprogramming} to a question purely about high-dimensional geometry.

As we explain next, we circumvent \Cref{q:heavy-reprogramming} by modifying the code so that list
recovery requires only $\Theta(\sqrt n)$ approximately known
symbols.

\subsection{Modifying the code and biasing the oracle}
We use a folded Reed--Solomon code $C\subseteq\Sigma^n$
of length $n$ and rate
$
R=\frac{1}{10\sqrt n}.
$
For any constant $c\in(0,1/2)$, define
$
\ell=\left\lceil 2^{n^c}\right\rceil.
$
The code satisfies the following low-agreement list-recovery property.
For every sequence of lists
\[
S_1,\ldots,S_n\subseteq\Sigma,
\qquad
|S_i|\leq\ell,
\]
there are at most
$
L=2^{O(n^c\log n)}
$
codewords \(x\in C\) such that $
x_i\in S_i$
on at least 
$\frac{\sqrt n}{5}$
coordinates.  Thus, the agreement required by list recovery exactly
matches the number of surviving heavy symbols supplied by the
query-norm polytope.

This lower-rate code solves the security mismatch, but it is not
immediately compatible with the Yamakawa--Zhandry quantum algorithm.
That algorithm moves to the Fourier basis and decodes an error relative
to the dual code \(C^\perp\).  For our choice of parameters, 
the dual code decoder tolerates errors supported on only \(O(\sqrt n)\) symbols.
Unfortunately, the error in the dual basis is nonzero in at least a constant fraction of coordinates, which makes decoding in the dual basis impossible and destroys completeness.

\paragraph{Using a biased random oracle.}
We repair correctness by biasing the random oracle\footnote{Note that our eventual certified randomness protocol will be in the ``standard'' quantum random oracle model.}.  Let \(p\) be the
unique power of \(1/2\) satisfying
\[
\frac{1}{80\sqrt n}
<
p
\leq
\frac{1}{40\sqrt n}.
\]
For every input \((i,x)\in[n]\times\Sigma\), independently sample
\[
H_i(x)=
\begin{cases}
1, & \text{with probability }p,\\
0, & \text{with probability }1-p.
\end{cases}
\]
The verifier continues to accept a codeword \(\bfx=(x_1,\ldots, x_n)\) only when $H_i(x_i)=0$ for every $i\in[n]$.
Since \(p\) is a negative power of two, this biased oracle can be
implemented from an ordinary random oracle: truncate the ordinary
oracle output to \(d\) bits, where \(p=2^{-d}\), and output \(1\) if
and only if all \(d\) bits are equal to \(1\).

As a result of this bias, the uniform superposition over the
zero-preimages of \(H_i\) is now highly concentrated on the zero
frequency after applying the Fourier transform.  The Fourier error is
nonzero in any particular coordinate with probability only
\(O(p)=O(1/\sqrt n)\).  Hence the total number of nonzero
coordinates is \(O(\sqrt n)\) with overwhelming probability, placing
the error within the unique-decoding radius of \(C^\perp\) which fixes the completeness problem described above.

\subsection{Security of the updated construction}
\label{sec:overview-biased-security}

\paragraph{Weighted counting under the biased oracle.}
The query-norm bounds, the swapping lemma, local stability, and the
query-norm polytope are pointwise statements about fixed oracles, so
they are unaffected by the bias.
Only the final counting steps in the reprogramming arguments must change.
Recall that if $h^S$ is obtained from $h$ by changing
$|S|$ oracle values from zero to one, then
\[
\frac{\Pr[H=h^S]}
     {\Pr[H=h]}
=
\left(\frac{p}{1-p}\right)^{|S|}.
\]
More generally, if \(T\) is a set of zero-valued inputs under $h$,
then the total probability mass of all reprogrammings of \(T\) compared to $\Pr[H=h]$ is
\begin{equation}
\label{eq:mass}
\frac{\sum_{S\subseteq T}\Pr[H=h^S]}{\Pr[H=h]} = 
\sum_{S\subseteq T}
\left(\frac{p}{1-p}\right)^{|S|}=
\left(\frac{1}{1-p}\right)^{|T|}.
\end{equation}

In the first reprogramming argument (\Cref{clm:claim1-intro}, outlined in \Cref{sec:claim1-intro}), the families of oracle-output
pairs are disjoint, and the adversary's probability of outputting
\(x\) decreases by at most a factor of two.
Taking $|T|=0.1n$, weighted counting therefore shows that
\[\Pr[H\text{ is bad\textsubscript{1}}]\leq \frac{2}{\alpha}(1-p)^{0.1n}\,.\]

In the second reprogramming argument (\Cref{clm:claim2-intro}, outlined in \Cref{sec:claim2-intro}), recall that we defined a bipartite graph where the left vertices are bad\textsubscript{2} oracles and the right vertices are all oracles. We then defined edges via the heavy-coordinate reprogramming lemma, and aimed to obtain a lower bound $D_\mathsf{L}$ on the degree of the left vertices and an upper bound $d_\mathsf{R}$ on the degree of the right. In the standard random oracle case, counting the number of edges in two ways then gives that the fraction of bad\textsubscript{2} oracles is at most $\frac{d_{\mathsf{R}}}{D_L}$, which, for a standard random oracle, coincides with $\Pr[H\text{ is bad\textsubscript{2}}]$. For our biased oracle, we should work directly with the quantity of interest, namely the probability mass on bad\textsubscript{2} oracles. 


Note that the total probability mass of all reprogrammings is larger than the probability mass of the original oracle by a factor of at least $
1/(1-p)^{4n/5}$, by \eqref{eq:mass}.
Since each resulting oracle belongs to at most \(L\) families, we have
\[
\begin{aligned}
\Pr[H\text{ is bad\textsubscript{2}}]\leq L(1-p)^{4n/5}
&=
\exp\bigl(-\Omega(\sqrt n)\bigr)
\end{aligned}
\]

\paragraph{From reprogramming to min-entropy.}
We can now put the two reprogramming arguments together. Suppose that the verifier accepts with probability at least \(\delta\), where \(\delta\) is inverse polynomial, but that the accepting output distribution has min-entropy at most \(h_\infty=o(n^c)\). Then some correct codeword \(\bfx\) has conditional probability at least \(2^{-h_\infty}\) among the accepting outputs, and hence unconditional output probability at least
$
\alpha
=
\delta \cdot 2^{-h_\infty}$.
In particular, $\alpha^{-1}
=
2^{o(n^c)}$.
Then, the first reprogramming argument (\Cref{clm:claim1-intro}) says that, except for an exceptional set of probability mass at most
$$
2\alpha^{-1}(1-p)^{0.1n},
$$
such a predictable correct answer must be heavily queried on almost all of its symbols. Since \(p=\Theta(1/\sqrt n)\) and \(c<1/2\), the factor \(\alpha^{-1}=2^{o(n^c)}\) is dominated by the \(\exp(-\Omega(\sqrt n))\) decay of \((1-p)^{0.1n}\), and hence this exceptional probability is negligible. 

The second reprogramming argument (\Cref{clm:claim2-intro}) says that the probability that any correct codeword is heavily queried on almost all of its symbols is at most
$$
L(1-p)^{4n/5}
=
\operatorname{negl}(n),
$$
where \(L=2^{O(n^c\log n)}\). Thus, except with negligible probability over the oracle, the two conclusions are incompatible. We conclude that every adversary making \(Q=2^{o(n^c)}\) queries and causing the verifier to accept with noticeable probability must produce an accepting output distribution with min-entropy at least \(h_\infty=o(n^c)\).

\paragraph{Returning to the standard QROM.}
Finally, although we have described the construction using a biased random
oracle \(H\), it can be implemented using an ordinary uniform random
oracle.  Write \(p=2^{-d}\), and let \(F\) be a uniformly random
\(d\)-bit-output oracle.  We may now define the biased oracle $H$ via
\[
H(u)=1
\quad\Longleftrightarrow\quad
F(u)=1^d,
\]
so that \(H(u)=1\) with probability exactly \(2^{-d}=p\), independently
for every input \(u\).

There is a final, minor subtlety: one might worry that the full oracle \(F\) reveals more
information than the single biased bit \(H\).  To handle this, we may
equivalently generate \(F\) by first sampling \(H\), then independently
sampling a residual oracle \(G\), and setting
\[
F(u)
=
\begin{cases}
1^d, & H(u)=1,\\
G(u), & H(u)=0,
\end{cases}
\]
where \(G(u)\) is uniform over
\(\{0,1\}^d\setminus\{1^d\}\).  For every fixed choice \(G=g\), an
adversary querying the full oracle \(F_{H,g}\) can be viewed as an adversary querying \(H\), with $g$ incorporated into its computation between queries.  Our security bounds hold for every such
fixing of \(g\). Applying the biased-oracle result after fixing \(g\), and
then averaging over \(G\), proves the min-entropy guarantee conditioned
on the entire uniform oracle \(F\).  Hence the construction is secure
in the standard QROM. We provide more details in Appendix~\ref{sec:app}.

\section{Preliminaries}\label[section]{sec:prelims}

The following lemma says that the Euclidean distance between two states upper-bounds the trace distance.
    \begin{lemma}\label[lemma]{thm:trace-dist-euclidian-dist}
        For any quantum pure states $\ket{\psi}$ and $\ket{\phi}$, \[\mathsf{TraceDist}\left(\ket{\psi}, \ket{\phi}\right) \leq \left\|\ket{\psi} - \ket{\phi}\right\|_2\,.\]
    \end{lemma}
    \begin{proof}
   This follows from the Fuchs--van de Graaf inequality, \textit{e.g.}, \cite[Ch. 9]{Nielsen2012}:
    \begin{align*}\mathsf{TraceDist}\left(\ket{\psi}, \ket{\phi}\right)\leq \sqrt{1-\abs{\braket{\psi}{\phi}}^2}
    \leq \sqrt{2\left(1-\abs{\braket{\psi}{\phi}}\right)}\leq \sqrt{2\left(1-\Re\braket{\psi}{\phi}\right)}=\left\|\ket{\psi}-\ket{\phi}\right\|_2.\tag*{$\qedhere$}\end{align*}
\end{proof}

\subsection{Quantum Random Oracle Model}\label[section]{sec:QROM}

Bellare and Rogaway \cite{BR93} defined the random oracle as a black-box oracle providing access to a uniformly random function whose inputs and outputs are arbitrarily long. The oracle $F$ takes as input a binary string $x \in \bit^*$ of any finite length, and outputs a uniformly random binary string $F(x)$ of any desired length. Furthermore, quantum algorithms and adversaries can query the random oracle on a quantum superposition of inputs, and the oracle responds coherently. This is known as the quantum random oracle model (QROM) \cite{BDFLSZ11}.

\paragraph{Biased Random Oracle.} 

Our construction uses a biased random oracle, which can in turn be constructed in the QROM. 
The biased random oracle $H$ is parametrized by a finite set $S$ and a (power of two) bias $p = 2^{-d}$ for some $d \in \bbN$. $H$ is a randomly chosen function mapping $S \to \bit$ that is sampled from the following distribution. For each input $x \in S$, independently sample $H(x) \in \bit$ such that $\Pr[H(x) = 1] = p$ and $\Pr[H(x) = 0] = 1-p$.

The biased random oracle $H$ can be implemented with query access to the regular random oracle $F$ by appropriately encoding any inputs and outputs of $F$. For instance, to query $H$ on some input $x \in S$, we represent $x$ as a bitstring, query $F(x)$, truncate the output to $d$ bits, and finally set $H(x) = 1$ if and only if $F(x) = 1^d$. This way, $\Pr[H(x) = 1] = \Pr[F(x) = 1^d] = 2^{-d} = p$.

In subsequent sections, this encoding is implicit, so we will treat the biased oracle $H$ as the random oracle without mentioning $F$.

\subsection{Query-Bounded Algorithms}\label[section]{sec:prelim-query-bounded-algorithms}
Here we define a quantum algorithm $\cA$ that makes $Q$ quantum queries to an oracle.

Let $\cA^\RO$ denote a quantum algorithm $\cA$ with quantum query access to oracle $\RO$. $\cA$ is parametrized by a number $Q \geq 1$, which is the maximum number of queries that $\cA$ can make. 

Without loss of generality, let $\cA^\RO$ operate as follows: 
\begin{itemize}
    \item $\cA$ starts with an initial pure state $\ket{\psi_0} = \ket{\psi^H_0}$ on registers $R = (R_Q, R_A)$
    
    \item For each query $j \in [Q]$:
    \begin{itemize}
        \item $\cA$ submits the query register $R_Q$ of $\ket{\psi_{j-1}^\RO}$ to the oracle $\RO$. Then $\RO$ acts as a phase oracle on $R_Q$, and returns $R_Q$ to $\cA$.
        \item $\cA$ applies $U_j$ to its state to obtain a new state $\ket{\psi_{j}^\RO}$.
    \end{itemize}
    \item Finally, $\cA$ applies measurement $M$ to $\ket{\psi_{Q}^\RO}$ and outputs the measurement outcome.
\end{itemize}

Next, we define the \textit{query norm} $\widetilde{w}_{i,x}^{j}(\RO)$ to upper-bound the amplitude magnitude given to query $(i,x)$ during the $j$-th query. In the literature on quantum query complexity, it is common to deal with the related quantity of \textit{query weight}, which is the probability of obtaining a query $(i,x)$ if we measure the $j$-th query. Query norm is the square root of query weight.

\begin{definition}[Query Norm]
For each possible classical query $(i,x) \in [n] \times \Sigma$ and each query index $j \in [Q]$, let $\widetilde{w}_{i,x}^{j}(\RO)$ be the \textbf{query norm} that $\cA^\RO$ gives to $(i,x)$ on the $j$-th query:
\[\widetilde{w}_{i,x}^{j}(\RO) = \left\|\Pi_{i,x} \ket{\psi_{j\smash{-1}}^\RO}\right\|_2\]
where $\Pi_{i,x}$ is the projector to the space querying the $(i,x)$ location of the oracle, i.e., \[\Pi_{i,x} = \ketbra{i,x}_{R_Q} \otimes \mathbb{I}_{R_A}.\] 
We define the \textbf{total query norm} that $\cA^{\RO}$ gives to $(i,x)$ as 
\[\widetilde{w}_{i,x}(\RO) := \sum_{j\in [Q]}\widetilde{w}^j_{i,x}(\RO).\]
\end{definition}

\begin{lemma}
\label{lemma:number_of_heavily_queried_coordinates}
The following holds with respect to any quantum query algorithm $\cA$ that makes $Q$ queries to an oracle $\RO$:
\begin{itemize}
    \item   For any $j \in [Q]$,
    \[\sum_{(i,x) \in [n] \times \Sigma} (\widetilde{w}^{j}_{i,x}(\RO))^2 = 1\;.\]
\item 
For any $\alpha>0$, there are at most $Q^2/\alpha^2$ inputs $(i,x)\in [n]\times \Sigma$ such that $\widetilde{w}_{i,x}(\RO)\ge \alpha$.
\end{itemize}
\end{lemma}
\begin{proof}
For the first item,
    \begin{align*}
        \sum_{(i,x) \in [n] \times \Sigma} (\widetilde{w}^{j}_{i,x}(\RO))^2 &= \sum_{(i,x) \in [n] \times \Sigma} \bra{\psi_{j-1}^{\RO}} \Pi_{(i,x)}  \ket{\psi_{j-1}^{\RO}} \\
        &= \bra{\psi_{j-1}^{\RO}} \cdot \left(\sum_{(i,x) \in [n] \times \Sigma} \Pi_{(i,x)} \right) \cdot \ket{\psi_{j-1}^{\RO}}\\
        &= \bra{\psi_{j-1}^{\RO}} \ket{\psi_{j-1}^{\RO}}\\
        &= 1\;.
    \end{align*}

For the second item, let 
    \[X = \{(i,x)\in [n]\times \Sigma : \widetilde{w}_{i,x}(\RO) \ge \alpha\}.\]
Then, for every $(i,x) \in X$, by Cauchy--Schwarz,
\[
\sum_{j=1}^{Q} \left(\widetilde{w}^{j}_{i,x}(\RO)\right)^2 
\geq \frac{1}{Q}\cdot \left(\sum_{j=1}^Q \widetilde{w}^{j}_{i,x}(\RO)\right)^2 \ge \frac{\alpha^2}{Q}
\]
Thus, summing over all $(i,x)\in X$ we get 
\[\sum_{(i,x)\in X} \sum_{j=1}^{Q} (\widetilde{w}^{j}_{i,x}(\RO))^2 \ge \frac{\alpha^2}{Q}\cdot |X|\;.\]
On the other hand, from the first item, we get
\begin{align*}
    \sum_{j=1}^{Q}\sum_{(i,x)\in X} (\widetilde{w}^{j}_{i,x}(\RO))^2 \le \sum_{j=1}^Q \sum_{(i,x) \in [n] \times \Sigma} (\widetilde{w}^{j}_{i,x}(\RO))^2 = Q\;.
\end{align*}
By comparing the two, we get $\frac{\alpha^2}{Q} \cdot |X|\le Q$, or equivalently, $|X|\le Q^2/\alpha^2$, which completes the proof.
\end{proof}
\Cref{thm:swapping-lemma} says that the adversary's states on two oracles $\RO$ and $\RO'$ will be close in Euclidean distance if the adversary gives small query norm to the positions where the oracles differ. Our version of this lemma is adapted from, though different from, \cite[Lemma 2]{Vaz98} and \cite[Theorem 3.3]{BBBV97}.
\begin{lemma}[Swapping Lemma, Adapted from 
    \cite{Vaz98} Lemma 2,
    \cite{BBBV97} Theorem 3.3]\label[lemma]{thm:swapping-lemma}
        Given two oracles $\RO, \RO'$, let $X \subseteq [n] \times \Sigma$ be the subset of inputs on which $\RO$ and $\RO'$ differ. Then, for any $j \in [Q]$,
        \[\left\|\ket{\psi^\RO_{j}} -  \ket{\psi^{\RO\smash{'}}_{j}}\right\|_2 \leq 
            2\sum_{j'=1}^{j} \sum_{(i,x) \in X} \widetilde{w}_{i,x}^{j'}(\RO)\,.\]
    \end{lemma}
\begin{proof}
    Write $\ket{\psi_j} := \ket{\psi_j^\RO}$ and $\ket{\phi_j} := \ket{\psi_j^{\RO'}}$ for brevity, and write $O$ and $O'$ for the phase oracles implementing $\RO$ and $\RO'$, respectively. Thus $O\ket{i,x}\ket{a} = (-1)^{\RO(i,x)}\ket{i,x}\ket{a}$, and likewise for $O'$. In particular, $O$ and $O'$ agree on every basis state whose query register lies outside $X$, while on every $(i,x)\in X$ they differ by a relative phase of $-1$, so
    \[
    (O-O') = \sum_{(i,x)\in X}\Bigl((-1)^{\RO(i,x)}-(-1)^{\RO'(i,x)}\Bigr)\Pi_{i,x}
    = \sum_{(i,x)\in X}\pm 2\,\Pi_{i,x}\,.
    \]
    Consequently, for any state $\ket{\psi}$,
    \begin{align*}
        \bigl\|(O-O')\ket{\psi}\bigr\|_2
        &= \Bigl\|\sum_{(i,x)\in X}\pm 2\,\Pi_{i,x}\ket{\psi}\Bigr\|_2
        \le 2\sum_{(i,x)\in X}\bigl\|\Pi_{i,x}\ket{\psi}\bigr\|_2\,,
    \end{align*}
    where 
    the last step is the triangle inequality. 

    The algorithm begins in the same state under both oracles, so $\ket{\psi_0}=\ket{\phi_0}$. For each subsequent query index $j\in[Q]$,
    \[
    \ket{\psi_j} = U_j O\ket{\psi_{j-1}}\qquad\text{and}\qquad
    \ket{\phi_j} = U_j O'\ket{\phi_{j-1}}\,.
    \]
    Therefore
    \begin{align*}
        \bigl\|\ket{\psi_j}-\ket{\phi_j}\bigr\|_2
        &= \bigl\|U_j O\ket{\psi_{j-1}} - U_j O'\ket{\phi_{j-1}}\bigr\|_2
        \\&= \bigl\|O\ket{\psi_{j-1}} - O'\ket{\phi_{j-1}}\bigr\|_2\\
        &\le \bigl\|O\ket{\psi_{j-1}} - O'\ket{\psi_{j-1}}\bigr\|_2
          + \bigl\|O'\ket{\psi_{j-1}} - O'\ket{\phi_{j-1}}\bigr\|_2\\
        &= 
          \bigl\|(O-O')\ket{\psi_{j-1}}\bigr\|_2
          +\bigl\|\ket{\psi_{j-1}}-\ket{\phi_{j-1}}\bigr\|_2\\
        &\le 
          \left(2\sum_{(i,x)\in X}\bigl\|\Pi_{i,x}\ket{\psi_{j-1}}\bigr\|_2 \right) + \bigl\|\ket{\psi_{j-1}}-\ket{\phi_{j-1}}\bigr\|_2\\
        &= 
          \left(2\sum_{(i,x)\in X}\widetilde{w}_{i,x}^{j}(\RO)\right) + \bigl\|\ket{\psi_{j-1}}-\ket{\phi_{j-1}}\bigr\|_2\,.
    \end{align*}
    Unrolling the recurrence and using $\ket{\psi_0}=\ket{\phi_0}$ gives the claimed bound.
\end{proof}

\section{Reprogramming Lemmas}\label[section]{sec:reprogramming-lemma}

In this section, we present a core technique for reprogramming a random oracle (\Cref{lemma:subcube with square root heavy coordinates}), which will be used in our proof of the min-entropy property (\Cref{sec:main}). \Cref{lemma:subcube with square root heavy coordinates} considers an oracle $H$ on which the adversary gives high $(\geq \beta)$ query norm to each input in a given set $M$. Essentially, the adversary spreads its query norm broadly over $M$. Next, we will reprogram $H$ on a subset of $M$ such that the adversary still spreads its norm somewhat broadly over $M$. We show that there are $2^{\Omega(|M|)}$ ways to reprogram $\RO$ such that the adversary gives somewhat high $\left(\geq \frac{\beta}{(8Q|M|)^2}\right)$ query norm to $\sqrt{|M|}$ inputs in $M$.

\subsection{Local Stability}


We start by a straightforward corollary to the swapping lemma (\Cref{thm:swapping-lemma}), which we call the ``local stability of query norms''. The idea is to view query norms themselves as a measurement on the state. Thus, changing input coordinates with small query norms barely changes the query norms of all other coordinates.
\begin{lemma}[Local Stability of Query Norms]\label[lemma]{lemma:local-stability-lemma}
Given two oracles $\RO, \RO'$ over $[n]\times \Sigma$, let $X \subseteq [n] \times \Sigma$ be the subset of inputs on which $\RO$ and $\RO'$ differ.
Then, for any $j \in [Q]$, and any $(i,x)\in [n]\times \Sigma$,
         \[\left|\widetilde{w}^j_{i,x}(\RO) - \widetilde{w}^j_{i,x}(\RO')\right| \le 2\sum_{j'<j} \sum_{(i',x') \in X} \widetilde{w}_{i',x'}^{j'}(\RO)\]
\end{lemma}
\begin{proof}
    From \Cref{thm:swapping-lemma}, we have that the $\ell_2$-distance between   $\ket{\psi^{\RO}_{j-1}}$ and $\ket{\psi^{\RO'}_{j-1}}$ is at most \[2\sum_{j'<j} \sum_{(i',x') \in X} \widetilde{w}_{i',x'}^{j'}(\RO),\] so it suffices to prove that 
    \[
    \left|\widetilde{w}^j_{i,x}(\RO) - \widetilde{w}^j_{i,x}(\RO')\right| \leq \left\|\ket{\psi^\RO_{j-1}} - \ket{\psi^{\RO'}_{j-1}}\right\|_2\;.
    \]
    Indeed, by the definition of query  norm, 
 \begin{align*}
 \left|\widetilde{w}^j_{i,x}(\RO) - \widetilde{w}^j_{i,x}(\RO')\right| &= \left|\big\|\Pi_{i,x}\ket{\psi^\RO_{j-1}}\big\|_2 - \big\|\Pi_{i,x}\ket{\psi^{\RO'}_{j-1}}\big\|_2\right|\\
 &\le 
 \big\|\Pi_{i,x}\ket{\psi^\RO_{j-1}} - \Pi_{i,x}\ket{\psi^{\RO'}_{j-1}}\big\|_2\tag{Reverse Triangle Inequality}\\
  &\le 
 \big\|\ket{\psi^\RO_{j-1}} - \ket{\psi^{\RO'}_{j-1}}\big\|_2\;.\tag*{\qedhere}
 \end{align*}
\end{proof}

Given an oracle $\RO$ and a subset of inputs $X \subseteq [n] \times \Sigma$, 
we denote by $\RO \oplus X$ the oracle that is obtained from $\RO$ by flipping the coordinates in $X$, i.e., for each $(i,x) \in [n]\times \Sigma$,
$(\RO \oplus X)(i,x) = \RO(i,x)$ if $(i,x) \notin X$ and 
$(\RO \oplus X)(i,x) = 1-\RO(i,x)$ if $(i,x) \in X$.

\begin{corollary}
Given any oracle $\RO$, any subset $X \subseteq [n] \times \Sigma$, and any input $(i,x)\in [n]\times \Sigma$,
we have
\begin{equation}\label{eq:local_stability}
|\widetilde{w}_{i,x}(\RO)- \widetilde{w}_{i,x}(\RO \oplus X)| \le 2Q
\cdot \sum_{(i',x')\in X} \widetilde{w}_{i',x'}(\RO) 
\end{equation}
\end{corollary}
\begin{proof}
\leavevmode
\vspace{-\baselineskip}
\setlength{\abovedisplayskip}{0pt}
\setlength{\abovedisplayshortskip}{0pt}
\begin{align*}
|\widetilde{w}_{i,x}(\RO)- \widetilde{w}_{i,x}(\RO \oplus X)|
 &=  \left|\sum_{j=1}^Q\left(\widetilde{w}^j_{i,x}(\RO)- \widetilde{w}^j_{i,x}(\RO \oplus X)\right)\right|\\
 &\le   \sum_{j=1}^Q \left|\widetilde{w}^j_{i,x}(\RO)- \widetilde{w}^j_{i,x}(\RO \oplus X)\right|\\
  &\le \sum_{j=1}^Q 2\sum_{j'<j} \sum_{(i',x')\in X} \widetilde{w}_{i',x'}^{j'}(\RO)\tag{\Cref{lemma:local-stability-lemma}}\\
 &\le 2Q\cdot \sum_{j'=1}^Q \sum_{(i',x')\in X}  \widetilde{w}_{i',x'}^{j'}(\RO)\\
 &= 2Q\cdot \sum_{(i',x')\in X} \widetilde{w}_{i',x'}(\RO)\;.\tag*{\qedhere}
\end{align*}
\end{proof}

\subsection{The Query Norm Polytope}

With \Cref{eq:local_stability}, the set of query norms that a quantum query algorithm can obtain on different oracles satisfies a system of linear inequality constraints. This is similar to a linear program, but without any objective function, thus we call it the \emph{query norm polytope}. 

In \Cref{sec:main}, we will consider an oracle $\RO$ for which the query algorithm places significant weight on each input in a set $M\subseteq [n]\times \Sigma$ of size $\Omega(n)$. The main result in this section is a proof that there are $2^{\Omega(|M|)}$ ways to reprogram $\RO$ by modifying a subset of the inputs in $M$, such that at least $\sqrt{|M|}$ of the inputs in $M$ still have significant query norm, and hence significant query weight.

For a set of inputs $M\subseteq[n]\times \Sigma$ and an oracle $\RO$, denote by 
\[\mathcal{F}_{\RO}(M)
    =\{\RO\oplus X : X\subseteq M\}\]
    the set of oracles obtained by reprogramming $\RO$ on any subset of $M$.

\begin{lemma}\label[lemma]{lemma:subcube with square root heavy coordinates}
    Let $M \subseteq [n]\times \Sigma$ be a non-empty set of inputs, let $\RO$ be an oracle, and let $\beta > 0$, such that for each $(i,x)\in M$, $\widetilde{w}_{i,x}(\RO)\ge \beta$.
    Then, for any $k \in \bbN$ such that $1+\binom{k}{2}\le |M|$, there exists a
    subset $A_k \subseteq M$ of size at most $\binom{k}{2}$
    such that for any oracle $\RO'\in \mathcal{F}_{\RO}(M\setminus A_k)$, there are at least $k$ inputs in $M$ for which $\cA^{\RO'}$ puts query norm $\geq \frac{\beta}{(8Q|M|)^2}$.
\end{lemma}

\begin{proof}
The proof is by induction on $k$.

Fix $\RO, n, \Sigma, \beta$. We prove by induction on $k$, that for all $M\subseteq[n]\times \Sigma$ of size $\ge 1+ \binom{k}{2}$, such that for all $(i,x)\in M$, $\widetilde{w}_{i,x}(\RO)\ge \beta$, there exists a
    subset $A_k \subseteq M$ of size at most $\binom{k}{2}$
    such that for any oracle $\RO'\in \mathcal{F}_{\RO}(M\setminus A_k)$, there are at least $k$ inputs in $M$ for which $\cA^{\RO'}$ puts query norm $\geq \frac{\beta}{(8Q|M|)^2}$.

\paragraph{\boldmath Base Case: $k=1$. Any reprogramming has one heavily queried input:} 
We begin with the special case of $k=1$. In this case $A_k = \emptyset$, and we need to prove that for any reprogramming obtained by flipping a subset $X\subseteq M$, i.e.\ for any $\RO'\in \mathcal{F}_{\RO}(M)$, at least one input in $M$ has query norm $\ge \frac{\beta}{(8Q|M|)^2}$.
In fact, we will show that at least one input in $M$ has query norm $\ge \frac{\beta}{4Q|M|}$, and we'll use the stronger lower bound later in the induction step as well.

Assume by contradiction that this is not the case. Then, there exists a set $X\subseteq M$, such that for all $(i,x)\in M$ we have 
\[
\widetilde{w}_{i,x}(\RO\oplus X) < \frac{\beta}{4Q|M|}
\]
But then, applying \Cref{eq:local_stability} with the role of $\RO$ and $\RO\oplus X$ swapped, on any $(i,x)\in M$ (at least one of which exists as $|M|\ge 1+\binom{k}{2}$), yields a contradiction, since
\begin{align*}
   \frac{3\beta}{4} \le |\widetilde{w}_{i,x}(\RO\oplus X) - \widetilde{w}_{i,x}(\RO) |\le 2Q\cdot \sum_{(i',x')\in X} \widetilde{w}_{i',x'}(\RO\oplus X) 
   \leq 2Q \cdot |X| \cdot \frac{\beta}{4Q|M|} \leq \frac{\beta}{2}\;.
\end{align*}
where the lower bound is by the hypothesis that $\widetilde{w}_{i,x}(\RO)\ge \beta$ for each $(i,x)\in M$.

\paragraph{Induction Step:}
Assume $k\ge 2$ and that the theorem is known for $k-1$.
Let $\tau := \frac{\beta}{(8Q|M|)^2}$.
There are two cases.

\medskip
\noindent\textbf{\boldmath Case 1: Every oracle in $\F_{\RO}(M)$ has at least $k$ inputs in $M$ with query norm at least
$\tau$.}
Then we are done immediately, by taking $A_k = \emptyset$.

\medskip
\noindent\textbf{\boldmath Case 2: Some oracle  in $\F_{\RO}(M)$ has fewer than $k$ inputs in $M$ with query norm at least $\tau$.}
Fix $\RO'$ to be such an oracle.
Define
\[
    \H:=\{(i,x)\in M: \widetilde{w}_{i,x}(\RO')\ge \tau\},
    \qquad
    \L:=M\setminus \H
\]
Since every oracle in $\mathcal{F}_{\RO}(M)$ has at least one input in $M$ with query norm $\ge \frac{\beta}{4Q|M|}$ (as we argued in the base case), we see that $\mathsf{Heavy}$ is non-empty, and by assumption its size is at most $k-1$.
Next, we claim that every oracle in $\F_{\RO}(\L)$ has at least one 
input in $\H$ whose query norm is at least $\tau$. 

Let $\RO_1$ be any oracle in $\F_{\RO}(\L)$. Our goal is to show that there exists $(i,x)\in \H$ such that $\widetilde{w}_{i,x}(\RO_1)\ge \tau$.

\paragraph{Hybrid Oracle.} Consider the hybrid oracle $\RO_2$ that equals $\RO_1$ on $\L$, equals $\RO'$ on $\H$, and equals both on $([n]\times \Sigma)\setminus M$.
We have that (i) $\RO_2 = \RO' \oplus X_1$ for some $X_1\subseteq \L$ on the one hand, and (ii) $\RO_2 = \RO_1 \oplus X_2$ for some $X_2\subseteq \H$ on the other hand.

\paragraph{\boldmath $\RO_2$ has a heavily queried input in $\H$.} Let $(i_0,x_0)\in \H$ be the input in $\H$ with the highest query norm according to $\cA^{\RO'}$. By the above discussion, its query norm is at least $\frac{\beta}{4Q|M|}$ (which is significantly larger than the threshold $\tau$ used to define $\H$).
From the first identity, by local stability (\Cref{eq:local_stability}), we get that its query norm on $\RO_2$ is at least 
\[\widetilde{w}_{i_0,x_0}(\RO_2) \ge \widetilde{w}_{i_0,x_0}(\RO') - 2Q\cdot \sum_{(i',x')\in X_1} \widetilde{w}_{i',x'}(\RO')
\ge \frac{\beta}{4Q|M|} - 2Q |M|\tau 
\ge \frac{\beta}{8Q|M|}\;.
\]

\paragraph{\boldmath $\RO_1$ has a heavily queried input in $\H$.} 
Next, assume towards contradiction that for every $(i,x)\in \H$, the query norm of $\cA^{\RO_1}$ on $(i,x)$ is less than $\tau$.
Then, for the above $(i_0,x_0)\in \H$ we get a contradiction from local stability (\Cref{eq:local_stability}) (using $\RO_2 = \RO_1 \oplus X_2$ for $X_2\subseteq \H$), since
\[
\widetilde{w}_{i_0,x_0}(\RO_2) \le \widetilde{w}_{i_0,x_0}(\RO_1) + 2Q\cdot  \sum_{(i,x)\in X_2} \widetilde{w}_{i,x}(\RO_1)
\le \tau + 2Q\cdot |\H|\cdot \tau \le   3Q|M|\cdot \tau < \frac{\beta}{8Q|M|}
\]

\paragraph{Applying the Induction Hypothesis.}
Observe that
\begin{align*}
    |\L| &=|M|-|\H|\\
    &\ge 1+\binom{k}{2}-(k-1)
    = 1 + \frac{k \cdot (k-1)}{2} - \frac{2 \cdot (k-1)}{2}\\
    &= 1 + \frac{(k-2) \cdot (k-1)}{2} = 1+\binom{k-1}{2}
\end{align*}
By the induction hypothesis applied to $\L$ and $k-1$, there is a set
$A_{k-1} \subseteq \L$ with
\[
    |A_{k-1}| \le \binom{k-1}{2}
\]
such that every oracle in $\F_{\RO}(\L\setminus A_{k-1})$ has at least $k-1$ inputs in $\L$ whose query norm is at least $\frac{\beta}{(8Q|\L|)^2} \ge \tau$.

But $\F_{\RO}(\L\setminus A_{k-1}) \subseteq \F_{\RO}(\L)$, and every oracle in $\F_{\RO}(\L)$ has at
least one additional input in $\H$ (which is disjoint from $\L$) that has query norm at least $\tau$. Hence, 
every oracle in $\F_{\RO}(\L\setminus A_{k-1})$ has at least $k$ inputs in $M$ whose query norm is at least $\tau$.

We take $A_k = A_{k-1}\cup \H$.
The number of coordinates in $A_k$ is at most $(k-1)+\binom{k-1}{2}
    =\binom{k}{2}$ and every oracle in $\F_{\RO}(M\setminus A_k)$ has at least $k$ inputs in $M$ whose query norm is at least $\tau$.
This completes the induction.
\end{proof}

    We remark that the only property about query norms that we used in the proof was local stability (\Cref{eq:local_stability}).

\section{Error-Correcting Codes}
In this section, we construct a family of error-correcting codes that will be used in our certifiable min-entropy construction. Our code family is similar to the one in \cite[Section~4]{YZ24}, except we adjust the parameters to allow list-recovery when only $O(\sqrt{n})$ symbols are approximately known.

We first define folded Reed--Solomon codes (\Cref{sec:ECC-definitions}). Then we construct a family of such codes by choosing the parameters appropriately (\Cref{sec:ECC-construction}). Finally, we state and prove the properties that we need from this code family (\Cref{sec:ECC-properties}). This section is adapted from \cite[Section~4]{YZ24} and uses similar definitions, notation, and proof strategies.

\subsection{Definitions}\label[section]{sec:ECC-definitions}
This section comes almost verbatim from \cite[Section~4]{YZ24}.\\

\noindent A code of length $n \in \mathbb{N}$ over a finite alphabet $\Sigma$ is a subset $C \subseteq \Sigma^n.$
\paragraph{Linear Codes.} A code $C$ is said to be linear if its alphabet is $\Sigma= \mathbb{F}_q$ for some prime power $q$ and $C \subseteq \mathbb{F}_q^n$ is a linear subspace of $\mathbb{F}_q^n$.
\paragraph{Folded Linear Codes.} A code $C$ is said to be a folded linear code if its alphabet is $\Sigma = \bbF_q^m$ for some prime power $q$ and a positive integer $m$ and $C \subseteq \Sigma^n$ is a linear subspace of $\bbF_q^{n\cdot m}$ where $n$ is the length of $C$ and we embed $C$ into $\bbF_q^{n\cdot m}$ in the canonical way. 

\paragraph{Dual Codes.} For a folded linear code $C \subseteq \Sigma^n$ over the alphabet $\Sigma = \bbF_q^m$, its dual code $C^\bot$ is defined as
\[C^\bot = \{\bfz \in \Sigma^n : \bfx \cdot \bfz = 0 \quad \forall \bfx \in C\}\]

\begin{definition}[List Recovery]\label[definition]{def:list-recovery}
Let \(\Sigma\) be a finite alphabet, let \(n\in\mathbb{N}\), and let
\(C\subseteq\Sigma^n\) be a code of  length $n$.  Let
\(\zeta\in[0,1]\), and let \(\ell,L\in\mathbb{N}\).
The code \(C\) is \textbf{\((\zeta,\ell,L)\)-list recoverable} if the
following holds:

For every sequence of sets
\[
  S_1,\ldots,S_{n}\subseteq\Sigma
  \qquad\text{such that}\qquad |S_i|\leq\ell
\qquad\forall i \in [n],
\]
there are at most \(L\) codewords
\(\bfx=(x_1,\ldots,x_{n})\in C\) satisfying
\[
  \left|\{\,i\in[n]:x_i\in S_i\,\}\right|
  \geq (1-\zeta)\cdot n.
\]
\end{definition}

\paragraph{Generalized Reed--Solomon Codes.} A generalized Reed--Solomon code $\mathsf{GRS}_{_{\mathbb{F}_q, \gamma,k, \bf{v}}}$ over $\mathbb{F}_q$ w.r.t. a generator $\gamma$ of $\mathbb{F}_q^*$, the degree parameter $0 \leq k\leq N$, and $\bfv = (v_1, \dots, v_N)\in {\mathbb{F}_q^*}^N$, where $N=q-1$ is defined as follows:
\[
\mathsf{GRS}_{_{\mathbb{F}_q, \gamma,k, \bf{v}}}:= \left\{\left[(v_1 f(\gamma), v_2f(\gamma^2),\dots , v_N f(\gamma^N)\right]:f \in \mathbb{F}_q[x]_{\text{deg}\leq k}\right\}
\]
where $\mathbb{F}_q[x]_{\text{deg}\leq k}$ denotes the set of polynomials over $\mathbb{F}_q$ of degree at most $k$. We remark that $\mathsf{GRS}_{_{\mathbb{F}_q, \gamma,k, \bf{v}}}$ is a linear code over $\mathbb{F}_q$ with length $N=q-1$ and dimension $k+1$. A Reed--Solomon
code is a special case of a generalized Reed--Solomon code where $\mathbf{v} = (1, \dots, 1)$.

There is a classical deterministic unique decoding algorithm $\mathsf{GRSUniqueDecode}_{\mathbb{F}_q, \gamma, k, \bf{v}}$ for $\mathsf{GRS}_{_{\mathbb{F}_q, \gamma,k, \bf{v}}}$ (e.g., the Berlekamp--Welch algorithm) that corrects up to $\lfloor (d_{\min}-1)/2 \rfloor = \lfloor (N-k-1)/2 \rfloor$ errors in time $\mathrm{poly}(N,\log q)$ \cite{Berlekamp68,BW86}.%
\footnote{One may clear the nonzero multipliers $v_i$ and apply any unique decoder for (ordinary) Reed--Solomon codes.} More precisely, for any $\textbf{z}\in \mathbb{F}_q^N$, if there exists $\textbf{x} \in \mathsf{GRS}_{_{\mathbb{F}_q, \gamma,k, \bf{v}}}$ with $\hw(\textbf{z}-\textbf{x})\le \lfloor (N-k-1)/2 \rfloor$, then $\mathsf{GRSUniqueDecode}_{\mathbb{F}_q, \gamma, k, \bf{v}}(\textbf{z})$ returns this (necessarily unique) codeword; otherwise it returns $\bot$.

\paragraph{Folded Reed--Solomon Codes.} A folded Reed--Solomon Code $\mathsf{FRS}$ is parametrized by a prime power $q$, a generator $\gamma$ of $\bbF_q^*$, the degree parameter $0 \leq k < N$, where $N= q-1$, and the folding parameter $m \in \bbN$, where $n = \tfrac{N}{m} \in \bbN$. Then the code $\mathsf{FRS}_{\mathbb{F}_q, \gamma, m,k}$ is defined as follows:
\[
\mathsf{FRS}_{\mathbb{F}_q, \gamma, m,k}= \left\{\left[f(\gamma^{jm+1}), f(\gamma^{jm+2}), \dots, f(\gamma^{jm+m})\right]_{j=0}^{n-1}: f \in \mathbb{F}_q[x]_{\mathsf{deg}\leq k}\right\}
\]

Note that $\Sigma= \mathbb{F}_q^m$ and $\mathsf{FRS}_{\mathbb{F}_q, \gamma, m,k} \subseteq \Sigma^n$,  $|\mathsf{FRS}_{\mathbb{F}_q, \gamma, m,k}|= q^{k+1}$, and $\mathsf{Rate}(\mathsf{FRS}_{\mathbb{F}_q, \gamma, m,k})= \frac{k+1}{N}$.

\subsection{Construction}\label[section]{sec:ECC-construction}

Here we construct a family of folded Reed--Solomon codes that will satisfy the requirements of \Cref{lem:suitablecodes}. Our construction simply entails choosing the parameters $(q, \gamma, m, k)$.

We construct the family of codes $\{C_\secp\}_{\secp \in \bbN}$ as follows. For all sufficiently large $\secp \in \bbN$, choose the following parameters:
\begin{itemize}
    \item Choose $c$ to be a constant $\in \left(0,\frac{1}{2}\right)$.
    \item $q(\secp) =2^{2\lfloor\log \lambda \rfloor}$
    \item $\gamma(\secp)$ is an arbitrary generator of $\mathbb{F}_q^*$. 
    \item $m(\secp) =2^{\lfloor\log \lambda \rfloor}+1$
    \item $N(\secp) = q(\secp)-1$
    \item $n(\secp) = \frac{N(\secp)}{m(\secp)} =2^{\lfloor\log \lambda \rfloor}-1$
    \item $R(\secp)= \frac{1}{10\sqrt{n(\lambda)}}$
    \item $k(\secp) =\lfloor R(\lambda)N(\lambda)\rfloor$
\end{itemize}
Note that $m(\lambda)=n(\lambda)+2, N(\lambda)=m(\lambda)n(\lambda)=n(\lambda)(n(\lambda)+2), q=(n(\lambda)+1)^2, \Sigma= \mathbb{F}_q^m$.\\

\noindent Finally, set $C_{\lambda}= \mathsf{FRS}_{\mathbb{F}_{q(\secp)}, \gamma(\secp), m(\secp), k(\secp)}$.

\subsection{Properties}\label[section]{sec:ECC-properties}
The following lemma gives several properties of our code that will be useful in the certifiable min-entropy protocol.
\begin{lemma}[Adapted from Lemma 4.2 in \cite{YZ24}]\label[lemma]{lem:suitablecodes}
    The family $\{C_{\lambda}\}_{\lambda \in \mathbb{N}}$ from \Cref{sec:ECC-construction} is a family of folded linear codes over the alphabet $\Sigma_{\lambda}= \mathbb{F}_{q(\lambda)}^{m(\lambda)}$ of length $n$ where $n = \Theta(\lambda), |\Sigma|= n^{\Theta(n)}= 2^{\Theta(n \log n)}$ that satisfies the following for sufficiently large $\secp$:
    \begin{enumerate}
        \item \label{property 1} $C_{\lambda
        }$ is $\left(\zeta(\secp), \ell(\secp), L(\secp)\right)$-list recoverable (\Cref{def:list-recovery}), where 
        \[\zeta(\secp) = 1-\frac{1}{5 \sqrt{n(\secp)}}, 
        \qquad \ell(\secp) = \left\lceil2^{\left(\lambda^c\right)}\right\rceil,
        \qquad L(\secp) = 2^{O(\lambda^c \log \lambda)}\;.\]
    \item \label{property 2}
    There is an efficient deterministic decoding algorithm $\mathsf{Decode}_{C_{\lambda}^\perp}$ for $C_\lambda^\perp$ that satisfies the following. 
    Let $\mathcal{D}$ be any distribution over $\Sigma$ that takes $\bf{0}$ with probability at least $1-\frac{1}{40\sqrt{n}}$. 
    Then, it holds that
    \[
    \Pr_{\bf{e}\leftarrow \mathcal{D}^{n}}[\forall \bfx \in C_{\lambda}^{\perp}, \mathsf{Decode}_{C_{\lambda}^{\perp}}(\bf{x}+\bf{e})=x]\geq 1-2^{-\Omega(\sqrt{\lambda})}
    \]

    \end{enumerate}
\end{lemma}
The rest of this subsection is devoted to proving \Cref{lem:suitablecodes}. 


\begin{proof}$ $
\paragraph{First Item.}

We prove Item \ref{property 1} of Lemma \ref{lem:suitablecodes}, relying on a result from \cite{GR08} and \cite{Rud07}.
\begin{theorem}
[\protect{\cite[Section~3.6]{Rud07}}, \cite{GR08}, \protect{\cite[Lemma~4.3]{YZ24}}]\label[theorem]{thm:GRlistrecovery_v2}

    Let $q$ be a prime power, $\gamma \in \bbF_q^*$ be a generator, $N=q-1$, $k < N$ be a positive integer, and $m$ be a positive integer that divides $N$. For positive integers $\ell,r$, and $s$, with $s \leq m$, and a real $0 < \zeta <1$,
    suppose 
    \begin{equation}\label{agreementineq}
        (1-\zeta) \cdot \frac{N}{m} \ge \left(1+\tfrac{s}{r}\right)\cdot \frac{\left(N\ell k^s\right)^{\frac{1}{s+1}}}{m-s+1}
    \end{equation}
    and 
    \begin{equation}\label{rootineq}
    (r+s)\cdot \left(\frac{N\ell}{k}\right)^{\frac{1}{s+1}}<q
    \end{equation}
    Then, $\mathsf{FRS}_{\mathbb{F}_q, \gamma,m,k }$ is $(\zeta, \ell, q^s)$-list recoverable.
\end{theorem}

We proceed with the proof of Item \ref{property 1} by establishing \Cref{agreementineq} and \Cref{rootineq}.
We begin with \Cref{agreementineq} which by rearranging is equivalent to 
    \begin{equation}\label{agreementineq_rearranged}
        \left(1+\tfrac{s}{r}\right)\cdot \tfrac{m}{m-s+1} \cdot \left(N\ell k^s\right)^{\frac{1}{s+1}}\le (1-\zeta) \cdot N 
    \end{equation}

Let $s = \lfloor{10\lambda^c\rfloor}$ and $r = 10s$.
Note that the right-hand side of \Cref{agreementineq_rearranged}  is \[(1-\zeta)\cdot N = \frac{N}{5\sqrt{n}} \ge 2k\]
So it remains to prove that the left-hand side of \Cref{agreementineq_rearranged} is at most $2k$.
Indeed, the left-hand side (of \Cref{agreementineq_rearranged}) is a product of three factors: 
\begin{itemize} 
\item[(i)] $\left(1+\frac{s}{r}\right) = 1.1$,
\item[(ii)] $\frac{m}{m-s+1}\le 1.1$, where the inequality holds for $\lambda$ large enough since $m\ge \lambda/2$ and $s\le 10\sqrt{\lambda}$,
\item[(iii)] and $(N\ell k^{s})^{\frac{1}{s+1}}$.
\end{itemize}
It remains to show that $(N\ell k^{s})^{\frac{1}{s+1}} \leq 1.1k$ to show that the left-hand side of \Cref{agreementineq_rearranged} is at most $(1.1)^3 k\le 2k$.
We do so by comparing $(N\ell k^{s})^{\frac{1}{s+1}}$ to $k$.
\begin{align*}
    \frac{(N\ell k^{s})^{\frac{1}{s+1}}}{k} = \left(\frac{N\ell}{k}\right)^{\frac{1}{s+1}} \le (N\ell)^{\frac{1}{s+1}} \le   2^{\lambda^c/(10\lambda^c)} \cdot 2^{2\log(\lambda)/(10\lambda^c)}\le 2^{0.1+o_{\lambda}(1)} \le 1.1
\end{align*} for large enough $\lambda$.
Overall, we got \[
\left(1+\tfrac{s}{r}\right)\cdot \tfrac{m}{m-s+1} \cdot \left(N\ell k^s\right)^{\frac{1}{s+1}} \le (1.1)^3 \cdot k \le 2k \le (1-\zeta) \cdot N
\]
for large enough $\lambda$, which completes the proof of \Cref{agreementineq}.

As for \Cref{rootineq},
following a similar calculation,
\begin{align*}(r+s)\cdot \left(\frac{N\ell}{k}\right)^{\frac{1}{s+1}} \le (r+s)\cdot 1.1 \le 130 \sqrt{\lambda}
\end{align*}
for large enough $\lambda$, and $q  \ge \lambda^2/4$ is asymptotically much larger.

In summary, for our choice of parameters, both \Cref{agreementineq} and \Cref{rootineq} are satisfied and therefore \Cref{thm:GRlistrecovery_v2} guarantees that the corresponding $m$-folded Reed--Solomon code is $(\zeta, \ell, L)$-list recoverable for $L=q^s \le (\lambda^2)^{10\lambda^c} = 2^{O(\lambda^c \log \lambda)}$, which completes the proof of Item~\ref{property 1}.

\paragraph{Second Item.}
We prove the dual-decoding property of
Lemma~\ref{lem:suitablecodes}.

Throughout this proof, we identify
\[
\Sigma^n=(\mathbb F_q^m)^n
\quad\text{with}\quad
\mathbb F_q^N
\]
by concatenating the \(n\) folded blocks. This is possible because $N = m \cdot n$. Accordingly, we use the
same notation for a folded word and its corresponding unfolded word.
In particular, when \(\bfe\leftarrow\mathcal D_r^n\), we also regard
$\bfe$ as an element of \(\mathbb F_q^N\), and
$\hw(\bfe)$ denotes its Hamming weight in
$\mathbb F_q^N$.\\
Since $C=\FRS_{\mathbb{F}_q,\gamma,m,k}$ is an \(m\)-folded
Reed--Solomon code, its dual is an \(m\)-folded generalized
Reed--Solomon code
\[
  C^\perp
  =
  \GRS_{\mathbb{F}_q,\gamma,N-k-2,\bfv}^{(m)}
\]
for some $\bfv \in \bbF_q^N$.
 Let $ d:=N-k-2.$

Thus \(d\) is a polynomial degree bound for the unfolded dual code.
Then the dimension of the dual code is
\[
  d+1=N-k-1,
\]
and its minimum distance (from the Singleton bound) is
\[
  d_{\min}(C^\perp)
  =
  N-d
  =
  k+2.
\]

Define
\[
  \tau
  :=
  \left\lfloor\frac{k+1}{2}\right\rfloor.
\]

Since our noise weight is at most $\tau=\lfloor(k+1)/2\rfloor$ with high probability (shown below), unique decoding of the unfolded GRS dual suffices.\footnote{By contrast, \cite{YZ24} uses list decoding of the dual~\cite{GR99} because their (unbiased) noise has weight about $N/2$, which exceeds the unique-decoding radius of their dual.} We define \(\Decode_{C^\perp}\) as follows:
\begin{description}
\item[\(\Decode_{C^\perp}(\bfz)\):]
On input \(\bfz\in\Sigma^n\), run
$\GRSUniqueDecode_{\mathbb{F}_q,\gamma,d,\boldsymbol \bfv}(\bfz)$
with decoding radius $\tau$. If $\GRSUniqueDecode$ outputs a codeword of $C^\perp$, then return this codeword. Otherwise, if $\GRSUniqueDecode$ outputs $\bot$, then return $\mathbf{0}$.
\end{description}

Note that
\[
  2\tau
  \leq k+1
  <k+2
  =
  d_{\min}(C^\perp),
\]
so there can be at most one dual codeword $\bfx$ within distance \(\tau\)
of any received word $\bfz$, and $\tau$ is exactly the unique-decoding radius of the dual. The decoder is therefore deterministic and
runs in time polynomial in \(N\) and \(\log q\), hence polynomial in
\(\lambda\).

Next define the set

\[\GoodErrors
  :=
  \left\{
    \bfe\in\mathbb{F}_q^N:
      \hw(\bfe)\leq\tau
  \right\}.
\]
For any \(\bfx\in C^\perp\) and \(\bfe\in\GoodErrors\), the codeword \(\bfx\)
is the unique dual codeword at distance at most \(\tau\) from
\(\bfx+\bfe\).  Since \(\hw(\bfe)\le\tau\), the unique decoder returns \(\bfx\), and
therefore
\begin{equation}
\label{eq:deterministic-correctness}
  \Decode_{C^\perp}(\bfx+\bfe)=\bfx
  \qquad
  \text{for every }
  \bfx\in C^\perp,\ \bfe\in\GoodErrors.
\end{equation}


It remains to bound
\[
  \Pr_{\bf{e}\leftarrow \cD_r^n}
  [\bfe\notin\GoodErrors].
\]
Parse \(\bfe\in\Sigma^n\) as
\[
  \bfe=(e_1,\ldots,e_n),
  \qquad \bfe_i\in\mathbb{F}_q^m,
\]
and define the number of nonzero folded blocks by
\[
  Y
  :=
  \bigl|\{i\in[n]:\bfe_i\neq\zero\}\bigr|.
\]
Next, by definition of $\mathcal{D}, \forall i\in [n]$, independently,
\[\Pr[\bfe_i = \mathbf{0}] \ge  1-\frac{1}{40\sqrt{n}}\;.
\]
We define $X$ to be a random variable that stochastically dominates $Y$:
\[
  X\sim\operatorname{Bin}\left(n,p\right)
\text{ for }p = \frac{1}{40 \sqrt{n}}\qquad 
  \mu:=\mathbb E[X]=\frac{\sqrt n}{40}.
\]

Every nonzero folded block contributes at most \(m\) nonzero
unfolded field coordinates.  Hence, when \(Y\leq2pn\),
\[
  \hw(\bfe)
  \leq mY
  \leq m \cdot 2pn
  =2pN
  =\frac{RN}{2} \le \frac{k+1}{2}
\]
The weight is an integer, so
\[
  \hw(\bfe)
  \leq
  \left\lfloor\frac{k+1}{2}\right\rfloor
  =
  \tau.
\]
This implies that $\bfe\in\GoodErrors$. In summary,
\[
  Y\leq2pn
  \quad\Longrightarrow\quad
  \bfe\in\GoodErrors.
\]

Next, we will use the multiplicative Chernoff bound,
\[\Pr[X>(1+\delta)\mu]\leq e^{-\delta^2\mu/(2+\delta)}\quad \forall \delta\geq 0\]
Setting $\delta = 1$, the multiplicative Chernoff bound says the following.
\[
  \begin{aligned}
  \Pr_{\bf{e}\leftarrow \cD_r^n}
  [\bfe\notin\GoodErrors]
  &\leq
  \Pr[Y>2pn]
  \leq \Pr[X>2\mu]\\
  &\leq
  \exp(-\mu/3)
  =  \exp({-\Omega(\sqrt{n})})
  =  \exp({-\Omega(\sqrt{\lambda})})
  \end{aligned}
\]

\par\smallskip\noindent
Combining this inequality with
\eqref{eq:deterministic-correctness} proves Item \ref{property 2}.
\end{proof}

\section{Certifiable Min-Entropy Protocol}
In this section, we present constructions of certifiable min-entropy protocols and prove their correctness. We follow the template of \cite{YZ24}'s proof of quantumness, but with the following main differences. First, we use the family of error-correcting codes constructed in \Cref{sec:ECC-construction}, which allow list recovery when only $(1 - \zeta)n = O(\sqrt{n})$ symbols are approximately known. Second, we use a biased random oracle to reduce the noise in the dual basis.

In \Cref{sec:prelim-certifiable-min-entropy}, we define certifiable min-entropy protocols. In \Cref{sec:CR-construction}, we present two constructions of such protocols. \Cref{construction:CR-bounded-min-entropy} achieves $o(\secp^c)$ bits of min-entropy, and \Cref{construction:CR-arbitrary-min-entropy} uses complexity leveraging to provide any $\poly(\secp)$ bits of min-entropy. Finally, in \Cref{sec:CR-correctness}, we prove the correctness of these protocols.

\subsection{Definition}\label[section]{sec:prelim-certifiable-min-entropy}
A certifiable min-entropy protocol is an interactive proof system between a quantum prover and a classical verifier, where any prover that causes the verifier to accept with noticeable probability must be sampling their accepting proofs from a distribution with high min-entropy. In the random oracle model, the prover has quantum query access to the random oracle, and the min-entropy of the proof is computed after conditioning on the choice of oracle. That way, the proof will provide additional randomness beyond the randomness of the oracle. 

\Cref{def:certifiable-min-entropy} below recalls the definition of a certifiable min-entropy protocol from \cite{YZ24} (almost verbatim). We modify the min-entropy property to allow us to specify the maximum number of queries $Q$ and the min-entropy threshold $\minent$.

\begin{definition}[Certifiable Min-Entropy Protocol, adapted from \cite{YZ24} Definition 3.5]\label[definition]{def:certifiable-min-entropy}
    A (keyless, non-interactive, publicly verifiable) \textbf{certifiable min-entropy protocol} relative to a random oracle (not necessarily uniformly random) consists of the algorithms $(\Prove, \Verify)$ with the following syntax.

    \paragraph{Syntax.}
    \begin{itemize}
        \item $\Prove^\RO (1^\secp, 1^\minent) \to \pi$: This is a QPT algorithm that takes the security parameter $1^\secp$ and a min-entropy threshold $1^\minent$ as input. It makes $\poly(\secp, \minent)$ quantum queries to a random oracle $\RO$ (not necessarily uniformly random), and outputs a classical proof $\pi$.
        \item $\Verify^\RO (1^\secp, 1^\minent, \pi) \to x$: This is a deterministic classical polynomial-time algorithm that takes $1^\secp$, $1^\minent$, and a proof $\pi$ as input.
        It makes $\poly(\secp, \minent)$ queries to the random oracle $\RO$, and outputs either a string $x$ (whose length may depend on $\secp, \minent$), or $\bot$ indicating rejection.
    \end{itemize}
    
    We also require a certifiable min-entropy protocol to satisfy the following properties:

    \paragraph{Correctness.} For any $\minent = \minent(\secp)$, we have
    \[\Pr_{\RO}\left[\Verify^\RO (1^\secp, 1^\minent, \pi) = \bot : \pi \gets \Prove^\RO (1^\secp, 1^\minent)\right] \leq \negl(\secp).\]

    \paragraph{$(Q,\minent)$-Certifiable Min-Entropy.} Given a function $Q = Q(\secp)$ and a polynomially-bounded function $\minent = \minent(\secp)$, the protocol satisfies \textbf{$(Q,\minent)$-certifiable min-entropy} if for any unbounded-time adversary $\cA$ that makes at most $Q$ quantum queries to $\RO$, and any inverse-polynomial function $\delta(\cdot)$\footnote{Specifically, we quantify over all $e \in \bbN$, and set $\delta(\secp) = \secp^{-e}$.}, there is a negligible function $\negl(\cdot)$ such that the following holds. Let $\cA^\RO_\top (1^\secp, 1^\minent)$ be the distribution $\Verify^\RO[1^\secp, 1^\minent, \cA^\RO(1^\secp, 1^\minent)]$, conditioned on the output not being $\bot$. Then:
    \begin{align*}
        \Pr_{\RO}\left[\Pr\left[\Verify^\RO [1^\secp, 1^\minent, \cA^\RO(1^\secp, 1^\minent)] \neq \bot\right] \geq \delta(\secp) \land \Minent\left(\cA^\RO_\top (1^\secp, 1^\minent)\right) \leq \minent(\secp)\right] \leq \negl(\secp).
    \end{align*}
    If $(Q, \minent)$ are not specified, we say that the protocol satisfies \textbf{certifiable min-entropy} if for any polynomial functions $(Q, \minent)$\footnote{Specifically, we quantify over all $d \in \bbN$, and set $\minent(\secp) = \secp^{d}$.}, the protocol satisfies $(Q,\minent)$-certifiable min-entropy. 
\end{definition}

\subsection{Construction}\label[section]{sec:CR-construction}
In this section, we give constructions of certifiable min-entropy protocols in the QROM. 

\paragraph{Error-Correcting Code.} Let $\{C_{\lambda}\}_{\lambda \in \mathbb{N}}$ be the family of codes given in \Cref{sec:ECC-construction}. In the following, we omit $\secp$ from the subscript of $C_\secp$ since it is clear from context. We also use the notation defined in Lemma \ref{lem:suitablecodes}. The code has an alphabet $\Sigma= \mathbb{F}_q^m$ and satisfies the requirements of Lemma \ref{lem:suitablecodes} with an arbitrary $0 < c < 1/2$.

\paragraph{Biased Random Oracle.} The biased random oracle $H$ is a random variable over functions mapping $[n] \times \Sigma \to \bit$, and it is parametrized by the unique bias
$\frac{1}{80\sqrt{n}} < p \le \frac{1}{40\sqrt{n}}$ which is a power of $1/2$. Next, $H$ is sampled from the following distribution, known as $\cH$. 
For each \((i, x) \in [n] \times \Sigma\), independently sample $H(i,x) \in \bit$ such that $\Pr[H(i,x) = 1] = p$. Finally, $H$ can be implemented in the QROM as discussed in \Cref{sec:QROM}.

Additionally, for each $i \in [n]$, let $H_i : \Sigma \to \bit$ be the function $H(i, \cdot)$. Finally, for any $\bfx = (x_1, \ldots, x_n) \in \Sigma^n$, let us abuse notation and say that $H(\bfx) = \left[H_1(x_1), \ldots, H_n(x_n)\right]$.


\paragraph{Certifiable Min-Entropy Protocols.} Now we construct certifiable min-entropy protocols using the code $C$ and the biased oracle $H$ defined above. \Cref{construction:CR-bounded-min-entropy} provides $o(\secp^c)$ bits of min-entropy, whereas \Cref{construction:CR-arbitrary-min-entropy} provides any $\poly(\secp)$ bits of min-entropy.

\begin{construction}[Min-entropy = $o(\secp^{c})$]\label[construction]{construction:CR-bounded-min-entropy}
$ $
\paragraph{$\Prove^H(1^{\secp}, 1^{h_{\infty}}):$} 
\begin{enumerate}
    \item For each $i \in [n]$, generate the state 
\[
\ket{\phi_i} \propto \sum_{\bfe_i \in \Sigma : H_i(\bfe_i)=0}\ket{\bfe_i}
\]
This is done as follows. Generate a uniform superposition over all $\bfe_i \in \Sigma$, coherently evaluate $H_i(\bfe_i)$, and measure the output. If the measurement outcome is $0$, then we have succeeded in generating $\ket{\phi_i}$. Repeat the above procedure at most $\secp$ times. If we fail to
generate $\ket{\phi_i}$ within $\secp$ trials, then abort $\Prove$ and output $\pi = \bot$. 

\item Set 
\[
\ket{\phi}= \ket{\phi_1}\otimes \dots \otimes \ket{\phi_n}
\]

\item Generate the state
\[
\ket{\psi} \propto \sum_{\bfx \in C}\ket{\bfx}
\]
\item Apply $\mathsf{QFT}_{\Sigma}^n$ to both $\ket{\psi}$ and $\ket{\phi}$.%
\footnote{We use the definition of $\mathsf{QFT}_{\Sigma}$ given in \cite[Section~2.2]{YZ24}} 
At this point, we have the state, 
\[
\ket{\eta}= \mathsf{QFT}_{\Sigma}^{\otimes n} \ket{\psi}\otimes \mathsf{QFT}_{\Sigma}^{\otimes n}\ket{\phi}
\]
\item Let $U_{\mathsf{add}}$ and $U_{\mathsf{decode}}$ be the unitaries defined as follows for any $\bfx, \bfe \in \Sigma^n$:
\[
\bf{\ket{x}}\bf{\ket{e}}\xrightarrow {U_{\mathsf{add}}}\bf{\ket{x}}\bf{\ket{x+e}}\xrightarrow {U_{\mathsf{decode}}}\bf{\ket{x- \mathsf{Decode}_{C^{\perp}}(x+e)}}\bf{\ket{x+e}}
\]
where $\mathsf{Decode}_{C^{\perp}}$ is the decoder for $C^{\perp}$ as required in Item \ref{property 2} of Lemma \ref{lem:suitablecodes}. Then apply the following unitaries to $\ket{\eta}$ (from right to left): 
$$\left[I \otimes \left(\mathsf{QFT}_{\Sigma}^{-1}\right)^{\otimes n}\right] \cdot U_{\mathsf{decode}}\cdot U_{\mathsf{add}}$$
\item Measure the second register to obtain $\mathbf{x}\in \Sigma^n$, and output $\pi = \bfx$.
\end{enumerate}

\paragraph{$\mathsf{Verify}^H(1^{\secp},1^{h_{\infty}},\pi)$:} 
\begin{enumerate}
    \item If $\pi = \bot$, then output $\bot$ and halt. Otherwise, parse $\pi = \mathbf{x}= (x_1, \dots, x_n)$.
    \item If $\mathbf{x}\in C$, and $H_i(x_i)= 0$ for all $i \in [n]$, then output $\bf{x}$. Otherwise, output $\bot$.
\end{enumerate}
\end{construction}
\begin{theorem}\label[theorem]{thm:CR-bounded-min-entropy}
    \Cref{construction:CR-bounded-min-entropy} is a certifiable min-entropy protocol (\Cref{def:certifiable-min-entropy}) that satisfies syntax, correctness, and $(Q, h_\infty)$-certifiable min-entropy for any $Q = 2^{o(\secp^c)}$ and any $\minent(\secp) = o(\secp^{c})$, where $c<1/2$ is a constant parameter of the list-recoverable code (\Cref{sec:ECC-construction}).
\end{theorem}
\begin{proof}
    It is clear by inspection that the syntax property is satisfied. Next, \Cref{thm:CR-correctness} proves correctness, and \Cref{thm:CR-for-depth-bounded-adversaries} proves $(Q, h_\infty)$-certifiable min-entropy for any $Q = 2^{o(\secp^c)}$ and any $\minent(\secp) = o(\secp^{c})$.
\end{proof}

Next, \Cref{construction:CR-arbitrary-min-entropy} uses complexity leveraging on \Cref{construction:CR-bounded-min-entropy} to provide $h_\infty'$ bits of min-entropy for any $h_\infty' = \poly(\secp')$, where $\secp'$ denotes the security parameter. Essentially, we run \Cref{construction:CR-bounded-min-entropy} with a security parameter $\secp$ that is at least $h_\infty'^{2/c}$. Then the min-entropy guarantee of \Cref{construction:CR-bounded-min-entropy} implies that this construction provides at least $\left(h_\infty'^{2/c}\right)^{c/2} = h_\infty'$ bits of min-entropy.

\begin{construction}[Min-entropy = $\poly(\secp')$]\label[construction]{construction:CR-arbitrary-min-entropy}
    Let $(\Prove, \Verify)$ be the functions defined in \Cref{construction:CR-bounded-min-entropy}. Additionally, let 
    \begin{align*}
        \secp(\secp', h_\infty') &= \max\left\{\secp', h_{\infty}'^{2/c}\right\}\\
        h_\infty(\secp', h_\infty') &= \secp(\secp', h_\infty')^{c/2}
    \end{align*}
    where $c$ is the parameter of the code $C$. Next, we construct $(\Prove', \Verify')$ below:

    \paragraph{$\Prove'^H(1^{\secp'}, 1^{h_{\infty}'}):$} Compute $\secp = \secp(\secp', h_\infty')$ and $h_\infty = h_\infty(\secp', h_\infty')$, compute $\pi \gets \Prove^H(1^\secp, 1^{h_\infty})$, and output $\pi$.

    \paragraph{$\mathsf{Verify}'^H(1^{\secp'},1^{h_{\infty}'},\pi)$:} Compute $\secp = \secp(\secp', h_\infty')$ and $h_\infty = h_\infty(\secp', h_\infty')$, compute $\bfx = \Verify^H(1^\secp, 1^{h_\infty}, \pi)$, and output $\bfx$.
\end{construction}

\begin{theorem}\label[theorem]{thm:CR-arbitrary-min-entropy}
    \Cref{construction:CR-arbitrary-min-entropy} is a certifiable min-entropy protocol (\Cref{def:certifiable-min-entropy}) that satisfies syntax, correctness, and certifiable min-entropy (i.e. it satisfies $(Q', h_\infty')$-certifiable min-entropy for any $Q' = \poly(\secp')$ and any $h_\infty'({\secp'}) = {\secp'}^d$, where $d \in \bbN$).
\end{theorem}
\begin{proof}
    It is straightforward to verify that the syntax property is satisfied. We note that $\Prove'$ and $\Verify'$ have runtimes and query complexities that are $\poly(\secp', h_\infty')$. These algorithms run $\Prove^H(1^\secp, 1^{h_\infty})$ or $\Verify^H(1^\secp, 1^{h_\infty}, \pi)$, whose runtimes and query complexities are $\poly(\secp, h_\infty)$. Furthermore, $\secp$ and $h_\infty$ are $\poly(\secp', h_\infty')$, so the runtimes and query complexities of $\Prove'$ and $\Verify'$ are $\poly(\secp', h_\infty')$.

    Next, we claim that the correctness property is satisfied. For any $h_\infty' = h_\infty'(\secp')$, let $\secp = \secp(\secp', h_\infty')$ and $h_\infty = h_\infty(\secp', h_\infty') = \secp^{c/2}$. Then there exists a negligible function $\negl$ such that
    \begin{align*}
        &\Pr_{\RO}\left[\Verify'^\RO \left(1^{\secp'}, 1^{h_\infty'}, \pi\right) = \bot : \pi \gets \Prove'^\RO \left(1^{\secp'}, 1^{h_\infty'}\right)\right]\\
        &\quad\quad= \Pr_{\RO}\left[\Verify^\RO \left(1^{\secp}, 1^{h_\infty}, \pi\right) = \bot : \pi \gets \Prove^\RO \left(1^{\secp}, 1^{h_\infty}\right)\right]\\
        &\quad\quad\leq \negl(\secp)
    \end{align*}
    This follows from the description of \Cref{construction:CR-arbitrary-min-entropy} and the correctness of \Cref{construction:CR-bounded-min-entropy} (\Cref{thm:CR-correctness}).
    
    It just remains to argue that $\negl\left(\secp\left[{\secp'}, h_\infty'(\secp')\right]\right)$ is negligible in ${\secp'}$. Since $\negl$ is negligible, we have that for any positive integer $d$, there exists an integer $\secp_d$ such that for all $\secp > \secp_d$, $|\negl(\secp)| < \secp^{-d}$. Also note that $\secp\left[{\secp'}, h_\infty'(\secp')\right] \geq {\secp'}$ for all $\secp' \in \bbN$. Then for all ${\secp'} > \secp_d$, $\secp\left[{\secp'}, h_\infty'(\secp')\right] \geq \secp' > \secp_d$, and
    \[\left|\negl\left(\secp\left[{\secp'}, h_\infty'(\secp')\right]\right)\right| < \secp\left[{\secp'}, h_\infty'(\secp')\right]^{-d} \leq {\secp'}^{-d}\]
    Therefore, the function $\negl\left(\secp\left[{\secp'}, h_\infty'(\secp')\right]\right)$ is negligible in ${\secp'}$. This shows that \Cref{construction:CR-arbitrary-min-entropy} satisfies correctness.

    Finally, we claim that \Cref{construction:CR-arbitrary-min-entropy} satisfies $(Q', h_\infty')$-certifiable min-entropy for any $Q' = \poly({\secp'})$ and $h_\infty'({\secp'}) = {\secp'}^d$, where $d \in \bbN$. First, let us set the parameters. Next, for any ${\secp'}$, let 
    \[\secp = \secp({\secp'}, h_\infty') = \max\left\{\secp', \secp'^{2d/c}\right\}, \quad Q = 2^{\secp^{c/2}}, \quad h_\infty = \secp^{c/2}.\] 
    Note that $(\secp', h_\infty')$ can be computed efficiently from $(\secp, c, d)$ as follows. If $\frac{2d}{c} > 1$, then $\secp = \secp'^{\frac{2d}{c}}$, so $\secp' = \secp^{\frac{c}{2d}}$. Otherwise, if $\frac{2d}{c} \leq 1$, then $\secp = \secp'$. Next, we can compute $h_\infty' = \secp'^d$.
    
    Second, let $\cA'$ be an arbitrary adversary for $(\Prove', \Verify')$ that takes inputs $\left(1^{\secp'}, 1^{h_\infty'}\right)$ and makes at most $Q'({\secp'})$ queries to $H$. Let us use $\cA'$ and $(c,d)$ to construct an adversary $\cA$ for $(\Prove, \Verify)$. $\cA$ takes inputs $\left(1^\secp, 1^{h_\infty}\right)$, then computes $(\secp', h_\infty')$ from $(\secp, c, d)$ as described above, and finally runs $\cA'\left(1^{\secp'}, 1^{h_\infty'}\right)$ and outputs its result. Additionally, let $\cA'^\RO_\top \left(1^{\secp'}, 1^{h_\infty'}\right)$ be the distribution of $\Verify'^\RO [1^{\secp'}, 1^{h_\infty'}, \cA'^\RO(1^{\secp'}, 1^{h_\infty'})]$ conditioned on the output not being $\bot$. Note that $\cA'^\RO_\top \left(1^{\secp'}, 1^{h_\infty'}\right)$ is the same as $\cA^\RO_\top \left(1^\secp, 1^{h_\infty}\right)$, which is the distribution of $\Verify^\RO [1^\secp, 1^{h_\infty}, \cA^\RO(1^\secp, 1^{h_\infty})]$ conditioned on the output not being $\bot$.

    Third, let us argue that the certifiable min-entropy property of \Cref{construction:CR-bounded-min-entropy} applies to our choices of $Q, h_\infty, \cA$. Note that
    \begin{align*}
        Q(\secp) &= 2^{\secp^{c/2}} = 2^{o\left(\secp^c\right)}\\
        h_\infty(\secp) &= \secp^{c/2} = o(\secp^{c})
    \end{align*}
    Furthermore, the number of queries to $H$ made by $\cA(1^\secp, 1^{h_\infty})$ is at most
    \[Q'({\secp'}) = \poly({\secp'}) \leq 2^{{\secp'}^{c/2}} \leq 2^{\secp^{c/2}} = Q(\secp)\]
    for sufficiently large ${\secp'}$.
    
    Fourth, the certifiable min-entropy property of \Cref{construction:CR-bounded-min-entropy} (\Cref{thm:CR-for-depth-bounded-adversaries}) implies that for any inverse-polynomial function $\delta(\cdot)$, there is a negligible function $\negl(\cdot)$ such that
    \begin{align*}
        \negl(\secp) &\geq \Pr_{\RO}\left[\begin{array}{cc}
         & \Pr\left[\Verify^\RO [1^{\secp}, 1^\minent, \cA^\RO(1^{\secp}, 1^\minent)] \neq \bot\right] \geq \delta(\secp) \\
         & \land \Minent\left(\cA^\RO_\top (1^{\secp}, 1^\minent)\right) \leq h_\infty(\secp)
        \end{array}\right]\\
        &= \Pr_{\RO}\left[\begin{array}{cc}
         & \Pr\left[\Verify'^\RO [1^{{\secp'}}, 1^{h_\infty'}, \cA'^\RO(1^{{\secp'}}, 1^{h_\infty'})] \neq \bot\right] \geq \delta(\secp) \\
         & \land \Minent\left(\cA'^\RO_\top (1^{{\secp'}}, 1^{h_\infty'})\right) \leq h_\infty(\secp)
        \end{array}\right]
    \end{align*}
    
    Fifth, let us relate $\delta(\secp)$ to $\delta({\secp'})$ and $h_\infty(\secp)$ to $h_\infty'(\secp')$. $\delta(\secp) = \secp^{-e}$ for some $e \in \bbN$, so for sufficiently large ${\secp'}$, 
    \[\delta({\secp'}) \geq \delta(\secp).\]
    Additionally,
    \[h_\infty'(\secp') = \secp'^d \leq \max\left\{\secp'^{c/2}, \secp'^{d}\right\} = \secp^{c/2} = h_\infty(\secp)\]
    Then
    \begin{align*}
        \negl(\secp) &\geq \Pr_{\RO}\left[\begin{array}{cc}
         & \Pr\left[\Verify'^\RO [1^{{\secp'}}, 1^{h_\infty'}, \cA'^\RO(1^{{\secp'}}, 1^{h_\infty'})] \neq \bot\right] \geq \delta(\secp) \\
         & \land \Minent\left(\cA'^\RO_\top (1^{{\secp'}}, 1^{h_\infty'})\right) \leq h_\infty(\secp)
        \end{array}\right]\\
        &\geq \Pr_{\RO}\left[\begin{array}{cc}
         & \Pr\left[\Verify'^\RO [1^{{\secp'}}, 1^{h_\infty'}, \cA'^\RO(1^{{\secp'}}, 1^{h_\infty'})] \neq \bot\right] \geq \delta({\secp'}) \\
         & \land \Minent\left(\cA'^\RO_\top (1^{{\secp'}}, 1^{h_\infty'})\right) \leq h_\infty'({\secp'})
        \end{array}\right]
    \end{align*}
    Finally,
    $\negl(\secp) = \negl(\secp[{\secp'}, h_\infty'({\secp'})])$ is negligible in ${\secp'}$ because $\secp[{\secp'}, h_\infty'({\secp'})] \geq {\secp'}$. This completes the proof of $(Q', h_\infty')$-certifiable min-entropy.
\end{proof}


\subsection{Proof of Correctness}\label[section]{sec:CR-correctness}
\begin{theorem}[Correctness]\label[theorem]{thm:CR-correctness}
    \Cref{construction:CR-bounded-min-entropy} satisfies correctness (\Cref{def:certifiable-min-entropy}).
\end{theorem}
The rest of \Cref{sec:CR-correctness} is devoted to proving \Cref{thm:CR-correctness}.\\

 First, we show that Prove aborts with negligible probability. Let $M:= |\Sigma|$, and for all $H$, define 
\[
  T_i^{H_i}
  =
  \{\bfe_i\in\Sigma:H_i(\bfe_i)=0\},
\]
Observe that  $|T_i^{H_i}
|\sim \operatorname{Bin}(M,1-p)$, and $\mu_i=\mathbb{E}[|T_i^{H_i}|]= (1-p)\cdot M\geq\frac{3M}{4}$. For all $i \in [n]$, let 
\[
G_i :=\{H: |T_i^{H_i}|\geq \frac{M}{2}\} 
\]
and let $G = \bigwedge_{i} G_i$. 
Because $\frac{M}{2}\leq \frac{2}{3}\mu_i$, the lower tail multiplicative Chernoff bound implies that 
\begin{align*}
    &\Pr_{H}[\neg G_i]\\
    &= \Pr_H[|T_i^{H_i}|<\frac{M}{2}]\\
    &\leq \Pr_H[|T_i^{H_i}|<(1-\frac{1}{3})\mu_i]\\
    &\leq \exp{-\Omega(M)}
\end{align*}
A union bound implies that $\Pr_H[\neg G]\leq n \cdot \exp{-\Omega(M)}$.
Fix any $H \in G$. Then for every $i \in [n]$, the probability that all $\lambda$ attempts to prepare $\ket{\phi_i}$ fail is at most $2^{-\lambda}$. From a union bound, $\Pr[\Prove \text{ aborts }| H]\leq n \cdot 2^{-\lambda}.$
It follows that 
\begin{align*}
    \Pr_{H}[\Prove^H \text{ aborts }]\\
    &\leq \Pr_{H}[\neg G]+ \Pr_{H}[\Prove^H \text{ aborts} \wedge H \in G]\\
    &\leq n \cdot \exp{-\Omega(M)}+  n \cdot 2^{-\lambda}\\
    &\leq \negl(\lambda)
\end{align*}

\noindent Now we will use a similar analysis to that of \cite{YZ24}. The main technical lemma in \cite{YZ24} is the following.
\begin{lemma}[\protect{\cite[Lemma 5.1]{YZ24}}]\label[lemma]{lem:techlemmayz}
    Let $\ket{\psi}$ and $\ket{\phi}$ be quantum states on a quantum system over an alphabet $\Sigma= \mathbb{F}_q^m$ written as 
    \[
    \ket{\psi}=\sum_{\mathbf{x} \in \Sigma^n}V(\mathbf{x})\mathbf{\ket{x}}
    \]
    \[
    \ket{\phi}= \sum_{\mathbf{e} \in \Sigma^n}W(\mathbf{e})\mathbf{\ket{e}}
    \]
    Let $F: \Sigma^n \rightarrow \Sigma^n$ be a function. Let $\mathsf{Good}\subseteq \Sigma^n \times \Sigma^n$ be a subset such that for any $\mathbf{(x,e)}\in \mathsf{Good}$, we have $F\mathbf{(x+e)}= \mathbf{x}$. Let $\mathsf{Bad}$ be the complement of $\mathsf{Good}$, i.e. $\mathsf{Bad}= (\Sigma^n \times \Sigma^n )\setminus \mathsf{Good}$. Suppose that we have,
    \begin{equation}\label{YZeq27}
\sum_{(\mathbf{x,e})\in\mathsf{Bad} }\Big|\widehat{V}(\mathbf{x})\widehat{W}(\mathbf{e})\Big|^2 \leq \epsilon
    \end{equation}
    
    and 
    \begin{equation}\label{YZeq28}
    \sum_{\mathbf{z}\in \Sigma^n}\left|\sum_{\substack{(\mathbf{x,e})\in\mathsf{Bad}\\: \mathbf{x}+\mathbf{e}=\mathbf{z} }}\widehat{V}(\mathbf{x})\widehat{W}(\mathbf{e})\right|^2 \leq \delta
    \end{equation}
    Let $U_{\mathsf{add}}$ and $U_F$ be unitaries defined as follows:
    \[
\ket{\bfx}\ket{\bfe}\xrightarrow {U_{\mathsf{add}}}\ket{\bfx}\ket{\bfx+\bfe}\xrightarrow {U_{F}}\ket{\bfx- F(\bfx+\bfe)}\ket{\bfx+\bfe}\,.
    \]
    Then we have
    \[
    \Big(I \otimes (\mathsf{QFT}_{\Sigma}^{-1})^{\otimes n}\Big) \cdot U_{F} \cdot U_{\mathsf{add}} \cdot \Big(\mathsf{QFT}_{\Sigma}^{\otimes n} \otimes \mathsf{QFT}_{\Sigma}^{\otimes n}\Big) \cdot \Big(\ket{\psi}\otimes \ket{\phi}\Big) \approx_{\sqrt{\epsilon}+ \sqrt{\delta}}|\Sigma|^{n/2} \cdot \sum_{\mathbf{z} \in \Sigma^n}V(\bfz) \cdot W(\mathbf{z}) \cdot \Big(\ket{\mathbf{0}} \otimes \ket{\mathbf{z}}\Big)\,.
    \]
\end{lemma}

In this notation, for any vectors $\bfv, \bfw$ and any $d > 0$, $\bfv \approx_{d} \bfw$ means $\|\bfv - \bfw\|_2 \leq d$.

Now, let us instantiate the variables of \Cref{lem:techlemmayz}. For a fixed \(H\), define
\[
  T_i^{H_i}
  =
  \{\bfe_i\in\Sigma:H_i(\bfe_i)=0\},
  \qquad
  T^H
  =
  T_1^{H_1}\times\cdots\times T_n^{H_n}.
\]
Let us define several amplitude functions:
\begin{align*}
    V(\bfx) &= \begin{cases}
        \frac{1}{\sqrt{|C|}} &\text{ if } \bfx \in C\\
         0 &\text{ otherwise }
    \end{cases} &\forall \bfx \in \Sigma^n\\
    W_i^{H_i}(\bfe_i) &= \begin{cases}
        \frac{1}{\sqrt{\left|T_i^{H_i}\right|}} &\text{ if } \bfe_i \in T_i^{H_i}\\
        0 &\text{ otherwise }
    \end{cases} &\forall i \in [n], \bfe_i \in \Sigma\\
    W^H(\bfe) &= \begin{cases}
        \frac{1}{\sqrt{\left|T^H\right|}} &\text{ if } \bfe \in T^H\\
        0 &\text{ otherwise }
    \end{cases} &\forall \bfe \in \Sigma^n\\
    &= \prod_{i \in [n]} W_i^{H_i}(\bfe_i)\,.
\end{align*}
    Additionally, let $F= \mathsf{Decode}_{C^{\perp}}$, and assume for now that with overwhelming probability over the choice of $H$, there exist subsets $\Good, \Bad$ such that inequalities \eqref{YZeq27} and \eqref{YZeq28} are true for negligible functions $\epsilon$ and $\delta$ respectively. This assumption will be justified in \Cref{cla:conditions}.
    Then, \Cref{lem:techlemmayz} directly corresponds to the honest $\Prove$ algorithm and implies that an honestly generated proof $\pi$ will satisfy $V(\pi) \cdot W^H(\pi) > 0$, except with negligible probability. If $V(\pi) \cdot W^H(\pi) > 0$, then $\pi \in C$ and $H_i(\pi_i) = 0$ for all $i \in [n]$. This value of $\pi$ will be accepted by $\Verify$, so $\Verify$ will not output $\bot$. This proves correctness.
    
    It remains to prove that \Cref{YZeq27,YZeq28} are true for negligible functions $\epsilon$ and $\delta$.\\
 Let \(M:=|\Sigma|\). Let \(\mu_p\) denote the distribution on functions
\(h:\Sigma\to\{0,1\}\) defined by
\[
  \mu_p(h)
  :=
  (1-p)^{|h^{-1}(0)|}p^{|h^{-1}(1)|}.
\]
Thus, the original biased oracle
\(H=(H_1,\ldots,H_n)\) is obtained by sampling $H_1,\ldots,H_n $ independently from $\mu_p$.

Define the following events:
\[
  \forall i \in [n] \quad \cE_i:=\left\{T_i^{H_i}\neq\emptyset\right\},
  \qquad
  \cE:=\bigwedge_{i=1}^n \cE_i.
\]
$\cE_i$ is the event that at least one input to $H_i$ hashes to $0$, and $\lnot \cE_i$ is the event that $H_i(\bfe_i)=1$ for every
$\bfe_i\in\Sigma$. Then
\[
  \Pr[\lnot \cE]
  \leq
  \sum_{i=1}^n\Pr[\lnot \cE_i]
  =np^M
  =\operatorname{negl}(\lambda),
\]

Define the conditioned one-coordinate distribution \(\widetilde{\mu}_p\) as the distribution of $H_i$ conditioned on $\cE_i$. Formally, for any $h : \Sigma \to \bit$,
\[
  \widetilde{\mu}_p(h)
  :=
  \begin{cases}
    \displaystyle
    \frac{(1-p)^{|h^{-1}(0)|}p^{|h^{-1}(1)|}}{1-p^M},
      & h^{-1}(0)\neq\emptyset,\\[3mm]
    0, & h^{-1}(0)=\emptyset.
  \end{cases}
\]
Sampling $H$ conditioned on $\cE$ is equivalent to sampling
$H_1,\ldots,H_n $ independently from $ \widetilde{\mu_p}.$
Throughout the remainder of the proof, whenever \(H\) or \(H_i\) appears
in an expectation or probability, it is sampled from $\widetilde{\mu_p}^n$ or $\widetilde{\mu_p}$, respectively, unless
stated otherwise. Note that $\Pr[\cE] = 1 - \negl(\secp)$, so $\mu_p^n$ and $\widetilde{\mu_p}^n$ are statistically close.

Finally, define
\[
  p^\star
  :=
  \frac{p-p^M}{1-p^M}.
\]
Then
\[
  0\leq p^\star\leq p,
  \qquad
  1-p^\star=\frac{1-p}{1-p^M}.
\]
\begin{claim}[Biased analogue of \protect{\cite[Claim 6.4]{YZ24}}]
\label[claim]{cla:conditions}

With probability $1-\negl(\lambda)$ over
$H= (H_1, \dots, H_n)\sim\widetilde{\mu_p}^{n}$, there are subsets
$\Good\subseteq\Sigma^n\times\Sigma^n$ and $\Bad = (\Sigma^n \times \Sigma^n) \backslash \Good$ such that
\[
  \mathsf{Decode}_{C^{\perp}}(\bfx+\bfe)=\bfx
  \qquad
  \text{for all }(\bfx,\bfe)\in\Good,
\]
and we have 
\[
    \sum_{(\bfx,\bfe)\in\Bad}
    \left|\widehat V(\bfx)\widehat W^H(\bfe)\right|^2
  \leq \negl(\lambda)
\]
and
\[
    \sum_{\bfz \in \Sigma^n} \left|\sum_{\substack{(\bfx,\bfe)\in\Bad\\:\bfx + \bfe = \bfz}}
    \widehat V(\bfx)\widehat W^H(\bfe)\right|^2
  \leq  \negl(\lambda)\,.
\]
\end{claim}

\begin{proof}
 
By the definition of $V$,
\begin{equation}
\label{eq:hat-v}
  \widehat V(\bfx)
  =
  \begin{cases}
    \frac{1}{\sqrt{|C^\perp|}},&\bfx\in C^\perp,\\
    0,&\bfx\notin C^\perp.
  \end{cases}
\end{equation}
Define
\begin{align*}
  \GoodErrors
  &:=
  \left\{
    \bfe\in\Sigma^n:
    \forall \bfx\in C^\perp, \mathsf{Decode}_{C^{\perp}}(\bfx+\bfe)=\bfx
  \right\},\\
  \BadErrors
  &:=
  \Sigma^n\setminus\GoodErrors.
\end{align*}
 Since $p^\star\leq p \leq \frac{1}{40 \sqrt{n}}$, \Cref{lem:suitablecodes} implies that,
\begin{equation}
\label{eq:prob-bad-error}
  \Pr_{\bfe\leftarrow\cD_{p^\star}^n}
  [\bfe\in\BadErrors]
  \leq2^{-\Omega(\sqrt{\lambda})}.
\end{equation}
for any distribution $\cD_{p^\star}$ that samples $0$ with probability $1-p^{\star}$.

Set
\[
  \Good
  :=
  C^\perp\times\GoodErrors,
  \qquad
  \Bad
  :=
  (\Sigma^n\times\Sigma^n)\setminus\Good.
\]
By definition,
\[
  \mathsf{Decode}_{C^{\perp}}(\bfx+\bfe)=\bfx
  \qquad
  \text{for every }(\bfx,\bfe)\in\Good.
\]

Because $\widehat V(\bfx)=0$ for $\bfx\notin C^\perp$,
\begin{align}
  \sum_{(\bfx,\bfe)\in\Bad}
  \left|\widehat V(\bfx)\widehat W^H(\bfe)\right|^2
  &=
  \sum_{\bfe\in\BadErrors}
  \left|\widehat W^H(\bfe)\right|^2,
  \label{eq:first-reduction}\\
  \sum_{\bfz\in\Sigma^n}
  \left|
    \sum_{\substack{(\bfx,\bfe)\in\Bad\\:\bfx+\bfe=\bfz}}
    \widehat V(\bfx)\widehat W^H(\bfe)
  \right|^2
  &=
  \sum_{\bfz\in\Sigma^n}
  \left|
    \sum_{\substack{\bfx\in C^\perp,\\ \bfe\in\BadErrors\\:\bfx+\bfe=\bfz}}
    \widehat V(\bfx)\widehat W^H(\bfe)
  \right|^2.
  \label{eq:second-reduction}
\end{align}
By a standard averaging argument, it is enough to establish
\begin{align}
  \mathbb E_{H}
  \left[
    \sum_{\bfe\in\BadErrors}
    \left|\widehat W^H(\bfe)\right|^2
  \right]
  &\leq2^{-\Omega(\sqrt{\lambda})},
  \label{eq:first-expectation-goal}\\
  \text{and}\qquad \mathbb E_{H}
  \left[
    \sum_{\bfz\in\Sigma^n}
    \Bigg|
      \sum_{\substack{\bfx\in C^\perp,\\ \bfe\in\BadErrors\\:\bfx+\bfe=\bfz}}
      \widehat V(\bfx)\widehat W^H(\bfe)
    \Bigg|^2
  \right]
  &\leq2^{-\Omega(\sqrt{\lambda})}.
  \label{eq:second-expectation-goal}
\end{align}
Then Markov's inequality will complete the proof of \Cref{cla:conditions}.

\begin{claim}
\label[claim]{cla:symmetry}
Let $\pi:\Sigma\to\Sigma$ be a permutation.  Then the distributions of $H_i$ and $H_i \circ \pi$ (resp. $H$ and $H \circ \pi^{\otimes n}$) are identical.
\end{claim}

\begin{proofofclaim}
Composition with $\pi$ preserves the cardinality of the set of preimages of $\zero$. More concretely, for all $i \in [n]$
\begin{equation*}
  |\{\bfe_i \in \Sigma:H_i(\bfe_i)=\zero\}|
  =
  |\{\bfe_i \in \Sigma:(H_i\circ\pi)(\bfe_i)=\zero\}|.\qedhere
\end{equation*}
\end{proofofclaim}

\paragraph{Proof of \Cref{eq:first-expectation-goal}.}
We first prove the following claim.

\begin{claim}
\label[claim]{cla:local-fourier}
For every $i\in[n]$,
\begin{equation}
\label{eq:hat-w-zero}
  \mathbb E_{H_i}
  \left[
    \left|\widehat W_i^{H_i}(\zero)\right|^2
 \right]
  =
  1-p^{\star}.
\end{equation}

\end{claim}

\begin{proofofclaim}
\Cref{eq:hat-w-zero} is proven as follows:
\begin{align*}
  \mathbb E_{H_i}
  \left[
    \left|\widehat W_i^{H_i}(\zero)\right|^2
  \right]
  &=
  \mathbb E_{H_i}
  \left[
    \left|
      \frac1{\sqrt M}
      \sum_{\bfa\in\Sigma}W_i^{H_i}(\bfa)
    \right|^2
  \right]\\
  &=
  \frac{
    \mathbb E_{H_i}\left[\left|T_i^{H_i}\right|\right]
  }{
   M
  }\\
  &=
  \frac{1-p}{1-p^M}.\\
  &=1-p^{\star}\,.\tag*{\qedhere}
\end{align*}
\end{proofofclaim}

Note that the distribution $\cD^{(i)}$ on $\Sigma$ defined by
\begin{equation}\label{eq:dist-and-expectation}\cD^{(i)}(e) = \mathbb E_{H_i}
  \left[
    \left|\widehat W_i^{H_i}(e)\right|^2\right]
\end{equation}
is the same for all $i\in [n]$. Thus, we can denote it by $\cD$. Observe further that $\cD$ satisfies the condition of \Cref{lem:suitablecodes}, namely $\cD(0) = 1-p^\star$ for $p^{\star} \le \frac{1}{40\sqrt{n}}$.


For every $\bfe=(\bfe_1,\ldots,\bfe_n)\in\Sigma^n$,

\begin{equation}
\label{eq:hat-w-product}
  \widehat W^H(\bfe)
  =
  \prod_{i=1}^n\widehat W_i^{H_i}(\bfe_i).
\end{equation}
Combining \eqref{eq:dist-and-expectation},
\eqref{eq:hat-w-product}, and conditional independence yields
\begin{equation}
\label{eq:product-law}
  \mathbb E_{H}
  \left[
    \left|\widehat W^H(\bfe)\right|^2
  \right]
  =
  \cD^n(\bfe).
\end{equation}
Consequently, by \eqref{eq:prob-bad-error},
\[
\begin{aligned}
  \mathbb E_{H}
  \left[
    \sum_{\bfe\in\BadErrors}\left|\widehat W^H(\bfe)\right|^2
  \right]
  &=
  \sum_{\bfe\in\BadErrors}\cD^n(\bfe)\\
  &=
  \Pr_{\bfe\leftarrow\cD^n}
  [\bfe\in\BadErrors]\\
  &\leq2^{-\Omega(\sqrt{\lambda})}.
\end{aligned}
\]
This proves \eqref{eq:first-expectation-goal}.

\paragraph{Proof of Equation \eqref{eq:second-expectation-goal}.}
Define $B:\Sigma^n\to \mathbb{C}$ so that $\widehat{B}$ satisfies the following:
\[
  \widehat B(\bfe)
  =
  \begin{cases}
    1,&\bfe\in\BadErrors\\
    0,&\bfe\notin\BadErrors
  \end{cases}\,.
\]

\begin{claim}
\label[claim]{cla:convolution-identity}

\[
  \sum_{\bfz\in\Sigma^n}
  \left|
    \sum_{\substack{\bfx\in C^\perp,\\ \bfe\in\BadErrors\\:\bfx+\bfe=\bfz}}
    \widehat V(\bfx)\widehat W^H(\bfe)
  \right|^2
  =
  \sum_{\bfz\in\Sigma^n}
  \left|V(\bfz)\left(B*W^H\right)(\bfz)\right|^2.
\]
\end{claim}

\begin{proofofclaim}
For every $\bfz\in\Sigma^n$,
\[
\begin{aligned}
  \sum_{\substack{\bfx\in C^\perp,\\ \bfe\in\BadErrors\\:\bfx+\bfe=\bfz}}
  \widehat V(\bfx)\widehat W^H(\bfe)&=
  \sum_{\substack{\bfx,\bfe\in\Sigma^n\\:\bfx+\bfe=\bfz}}
  \widehat V(\bfx)
  \bigl(\widehat B(\bfe)\widehat W^H(\bfe)\bigr)\\
  &=
  \left[
    \widehat V*
    \left(\widehat B\cdot\widehat W^H\right)
  \right](\bfz)\\
  &=
  \reallywidehat{\,\left[V\cdot\left(B*W^H\right)\right]\,}(\bfz)\,.
\end{aligned}
\]
The first equality uses
$\widehat V(\bfx)=0$ for all $\bfx \notin C^\perp$; the last equality is the
convolution theorem.
Parseval's theorem completes the proof.
\end{proofofclaim}

\begin{claim}
\label[claim]{cla:uniform-convolution}
For every \(\bfz\in\Sigma^n\),
\[
  \mathbb E_{H}
  \left[
    \left|\left(B*W^H\right)(\bfz)\right|^2
  \right]
  \leq2^{-\Omega(\sqrt{\lambda})}.
\]
\end{claim}

\begin{proofofclaim}
First observe that for any $\bfz_0,\bfz_1\in\Sigma^n$,
\begin{equation}
\label{eq:translation-constant}
  \mathbb E_{H}
  \left[
    \left|\left(B*W^H\right)(\bfz_0)\right|^2
  \right]
  =
  \mathbb E_{H}
  \left[
    \left|\left(B*W^H\right)(\bfz_1)\right|^2
  \right].
\end{equation}
Indeed, let
\[
  \tau_i(\bfa):=\bfa+(\bfz_0)_i-(\bfz_1)_i
  \qquad\forall i\in[n], \bfa \in \Sigma
\]
and $\boldsymbol\tau=(\tau_1,\ldots,\tau_n)$. 
Therefore,
\[
\begin{aligned}
  \mathbb E_{H}
  \left[
    \left|\left(B*W^H\right)(\bfz_0)\right|^2
  \right]&=
  \mathbb E_{H}
  \left[
    \left|
      \sum_{\bfx\in\Sigma^n}
      B(\bfx)W^H(\bfz_0-\bfx)
    \right|^2
  \right]\\
  &=
  \mathbb E_{H}
  \left[
    \left|
      \sum_{\bfx\in\Sigma^n}
      B(\bfx)W^{H\circ\boldsymbol\tau}(\bfz_1-\bfx)
    \right|^2
  \right]\\
  &=
  \mathbb E_{H}
  \left[
    \left|
      \sum_{\bfx\in\Sigma^n}
      B(\bfx)W^H(\bfz_1-\bfx)
    \right|^2
  \right]\\
  &=
  \mathbb E_{H}
  \left[
    \left|\left(B*W^H\right)(\bfz_1)\right|^2
  \right],
\end{aligned}
\]
where the third equality uses
\Cref{cla:symmetry}.

Then, for any $\bfz\in \Sigma^n$,
\[
\begin{aligned}
  \mathbb E_{H}
  \left[
    \left|\left(B*W^H\right)(\bfz)\right|^2
  \right]
  &=
  \frac1{M^n}
  \sum_{\bfz'\in\Sigma^n}
  \mathbb E_{H}
  \left[
    \left|\left(B*W^H\right)(\bfz')\right|^2
  \right]\\
  &=
  \frac1{M^n}
  \mathbb E_{H}
  \left[
    \sum_{\bfz'\in\Sigma^n}\left|\left(B*W^H\right)(\bfz')\right|^2
  \right]\\
  &=
  \frac1{M^n}
  \mathbb E_{H}
  \left[
    \sum_{\bfe\in\Sigma^n}
    \left|
      M^{n/2}\widehat B(\bfe)\widehat W^H(\bfe)
    \right|^2
  \right]\\
  &=
  \mathbb E_{H}
  \left[
    \sum_{\bfe\in\BadErrors}
    \left|\widehat W^H(\bfe)\right|^2
  \right]\\
  &\leq e^{-\Omega(\sqrt{\lambda})}.
\end{aligned}
\]
The third equality uses Parseval's theorem and the convolution theorem; the
last inequality is from \eqref{eq:first-expectation-goal}.
\end{proofofclaim}

Finally, \Cref{cla:convolution-identity,cla:uniform-convolution}, together with the definition of
$V$, give
\[
\begin{aligned}
  \mathbb E_{H}
  \Bigg[
    \sum_{\bfz\in\Sigma^n}
    \Bigg|
      \sum_{\substack{\bfx\in C^\perp,\\ \bfe\in\BadErrors\\:\bfx+\bfe=\bfz}}
      \widehat V(\bfx)\widehat W^H(\bfe)
    \Bigg|^2
  \Bigg]&=
  \mathbb E_{H}
  \left[
    \sum_{\bfz\in\Sigma^n}
    \left|V(\bfz)\left(B*W^H\right)(\bfz)\right|^2
  \right]\\
  &=
  \frac1{|C|}
  \sum_{\bfz\in C}
  \mathbb E_{H}
  \left[
    \left|\left(B*W^H\right)(\bfz)\right|^2
  \right]\\
  &\leq e^{-\Omega(\sqrt{\lambda})}.
\end{aligned}
\]
This proves \Cref{eq:second-expectation-goal} and completes the proof of \Cref{cla:conditions}.
\end{proof}

\section{Proof of the Min-Entropy Property}\label[section]{sec:main}
Our main result (\Cref{thm:CR-for-depth-bounded-adversaries}) says that the construction in \Cref{sec:CR-construction} provides $\Omega(\secp^{c})$ bits of min-entropy, for some constant $c<1/2$, as long as the adversary is limited to $2^{o(\secp^c)}$-many queries.

\begin{theorem}[Certifiable Min-Entropy]\label[theorem]{thm:CR-for-depth-bounded-adversaries}
    \Cref{construction:CR-bounded-min-entropy} (which is relative to a biased random oracle) satisfies $(Q,\minent)$-certifiable min-entropy (\Cref{def:certifiable-min-entropy}) for any $Q = 2^{o(\secp^c)}$ and any $\minent(\secp) = o(\secp^{c})$, where $c<1/2$ is a constant parameter of the list-recoverable code (\Cref{sec:ECC-construction}).
\end{theorem}

As mentioned earlier, the biased random oracle used in \Cref{construction:CR-bounded-min-entropy} can be implemented in the standard QROM as described at the start of \Cref{sec:QROM}. Moreover, though perhaps not immediately obvious, the resulting protocol in the QROM retains exactly the same certifiable min-entropy guarantee. So, we have the following corollary.
\begin{corollary}\label[corollary]{cor:CR-for-depth-bounded-adversaries-qrom}
    \Cref{construction:CR-bounded-min-entropy}, instantiated in the QROM, satisfies $(Q,\minent)$-certifiable min-entropy (\Cref{def:certifiable-min-entropy}) for any $Q = 2^{o(\secp^c)}$ and any $\minent(\secp) = o(\secp^{c})$, where $c<1/2$ is a constant parameter of the list-recoverable code (\Cref{sec:ECC-construction}).
\end{corollary}
We include a proof of Corollary~\ref{cor:CR-for-depth-bounded-adversaries-qrom} in Appendix~\ref{sec:app}. The rest of \Cref{sec:main} is devoted to proving \Cref{thm:CR-for-depth-bounded-adversaries}.

\medskip
Let $(\RO, \bfX)$ be the random variables referring to the random oracle and $\cA$'s output respectively, and let $(\ro, \bfx)$ be generic values that they take. $(\RO, \bfX)$ are jointly sampled from the following distribution. First $\RO \gets \cH$, and then $\bfX \gets \cA^\RO(1^\secp, 1^{\minent})$.

\subsection{Low-entropy adversaries heavily query correct answers.}
First, we prove that if the adversary is able to output a correct answer $\bfx$ with high probability, then they must have given non-negligible query norm to most symbols of $\bfx$. Intuitively, this is because the adversary must check that most symbols of $\bfx$ actually hash to $\mathbf{0}$; otherwise, they would output $\bfx$ even when those symbols hash to $1$.
\begin{lemma}\label[lemma]{thm:must-heavy-query}
    Let $\alpha\in (0,1]$, $p \in (0,1)$, and $s \in [n]$. Let $Q$ be the query complexity of $\cA$.

        Recall that $\RO$ is a $\bias$-biased oracle. Then, over the randomness of $(\RO, \bfX)$, the probability that we sample a pair $(\ro, \bfx)$ that satisfies all of the following conditions is at most $2\cdot (1-\bias)^{s}$:
    \begin{enumerate}
        \item There exists a set $I \subseteq [n]$, $|I|=s$, such that for all $i \in I$,
        \[\widetilde{w}_{i, x_i}(\ro) < \frac{\alpha}{4 n}\,.\]

        \item $\bfx$ has high probability given $\ro$: $\Pr[\bfX = \bfx |\RO = \ro] \geq \alpha$.
        \item $\bfx$ is correct: $\bfx \in C$ and $\ro(\bfx) = \mathbf{0}$.
    \end{enumerate}
\end{lemma}

\begin{proof}
Let $S$ be the set of all pairs $(\ro, \bfx)$ that satisfy the conditions of \Cref{thm:must-heavy-query}. $\Pr[S]$ is the probability that all the conditions of \Cref{thm:must-heavy-query} are satisfied, over the randomness of sampling $\ro \gets \cH$ and $\bfx \gets \cA^\ro$. We will show that $\Pr[S] \leq 2\cdot (1-\bias)^s$.
    
    Next, for each $(\ro,\bfx) \in S$, let us construct a set $T_{\ro, \bfx}$ of pairs $(\ro', \bfx)$. 
    
    Fix a set $I \subseteq [n]$ of size $s$ such that for all $i \in I$, $\widetilde{w}_{i, x_i}(\ro) < \frac{\alpha}{4n}$.
    To construct $T_{\ro, \bfx}$, choose any (potentially empty) subset of $\{(i, x_i)\}_{i \in I}$ and reprogram $\ro$ on these inputs, flipping the output from $0$ to $1$. Let us call the new oracle $\ro'$. Let $T_{\ro, \bfx}$ be the set of all pairs $(\ro', \bfx)$ that can be constructed in this way. 
    
    \paragraph{Properties:} 
    \begin{itemize}
        \item The sets $T_{\ro, \bfx}$ are disjoint. That is to say, for any two distinct pairs $(h_1, \bfx_1)$ and $(h_2, \bfx_2)$ in $S$, 
    \[T_{h_1, \bfx_1} \cap T_{h_2, \bfx_2} = \emptyset\;.\]
    Indeed, for any $(\ro', \bfx') \in T_{\ro, \bfx}$, $\bfx = \bfx'$, and $\ro$ can be computed from $(\ro', \bfx')$ by reprogramming $\ro'$ to map $\bfx'$ to $\mathbf{0}$. This procedure computes $(\ro, \bfx)$ from any $(\ro', \bfx') \in T_{\ro, \bfx}$, so the sets $T_{\ro, \bfx}$ must be disjoint.

    \item 
        For any $(\ro', \bfx) \in T_{\ro, \bfx}$,
        \begin{equation}\label{eq:X is likely under h'}
        \Pr[\bfX = \bfx|\RO = \ro']\ge \frac{1}{2} \cdot \Pr[\bfX = \bfx|\RO = \ro] \,.
        \end{equation}
    This follows from the swapping lemma (\Cref{thm:swapping-lemma}) and \Cref{thm:trace-dist-euclidian-dist}. The adversary samples their output $\bfx$ by applying a measurement $M$ to their final state (either $\ket{\psi_Q^\ro}$ or $\ket{\psi_Q^{\ro'}}$).
%

    \begin{align*}
        \Pr[\bfX = \bfx|\RO = \ro] - \Pr[\bfX = \bfx|\RO = \ro'] &\leq \mathsf{TraceDist}\left[\ket{\psi_Q^{\ro'}}, \ket{\psi_Q^{\ro}}\right]\\
        &\leq \left\|\ket{\psi_Q^{\ro'}} -  \ket{\psi_Q^{\ro}}\right\|_2 \tag*{(\Cref{thm:trace-dist-euclidian-dist})}\\
        &\leq 2\sum_{j \in [Q]}\sum_{\substack{(i, x_i) : \\\ro(i, x_i) \neq \ro'(i, x_i)}} \widetilde{w}_{i, x_i}^j(\ro) \tag*{(\Cref{thm:swapping-lemma})}\\
        &\leq 2\cdot \frac{\alpha}{4n} \cdot |I|
        \\&\leq \frac{\alpha}{2}\\
        &\leq  \frac{1}{2} \cdot \Pr[\bfX = \bfx|\RO = \ro],
    \end{align*}
    which yields \Cref{eq:X is likely under h'} upon rearrangement.
    \end{itemize}

    \paragraph{Finishing the proof:} Now, we will show that $\Pr[S]\leq 2\cdot (1-\bias)^{s}$.
    
    Let $\Pr[T_{\ro, \bfx}]$ be over the randomness of sampling $\ro' \gets \cH$ and $\bfx \gets \cA^{\ro'}$ given $\ro'$. Since the sets $\left(T_{\ro, \bfx}\right)_{(\ro, \bfx) \in S}$ are mutually disjoint,
    \begin{align*}
        1 &\geq \sum_{(\ro, \bfx) \in S} \Pr[T_{\ro, \bfx}]\\
        &= \sum_{(\ro, \bfx) \in S} \sum_{(\ro', \bfx) \in T_{\ro, \bfx}} \Pr[\RO=\ro'] \cdot \Pr[\bfX = \bfx|\RO = \ro']\\
        &\geq \sum_{(\ro, \bfx) \in S} \sum_{(\ro', \bfx) \in T_{\ro, \bfx}} \Pr[\RO=\ro'] \cdot \frac{1}{2} \cdot \Pr[\bfX = \bfx|\RO = \ro]\tag{By \Cref{eq:X is likely under h'}}\\
        &= \sum_{(\ro, \bfx) \in S}  \Pr[\bfX = \bfx|\RO = \ro]\cdot \frac{1}{2} \cdot \sum_{(\ro', \bfx) \in T_{\ro, \bfx}} \Pr[\RO=\ro']\,.
        \end{align*}
        We compare the internal sum to $\Pr[H=h]$.
        Since $H$ is $\bias$-biased, $\Pr[H(i,x) = 1] = p$ for any $(i,x) \in [n] \times \Sigma$. If $h'$ is obtained from $h$ by reprogramming  a subset of coordinates $J \subseteq I$ from $0$ to $1$, then 
        \[\Pr[\RO=\ro'] = \Pr[\RO=\ro] \cdot \left(\tfrac{\bias}{1-\bias}\right)^{|J|}.\]
        Summing over all subsets $J \subseteq I$, we get 
        \begin{align*}
            \sum_{(\ro',\bfx) \in T_{\ro, \bfx}}\Pr[\RO=\ro'] &= \Pr[\RO=\ro] \cdot \sum_{J\subseteq I} \left(\tfrac{\bias}{1-\bias}\right)^{|J|}\\
            &= \Pr[\RO=\ro] \cdot \sum_{k = 0}^{s}\sum_{\substack{J\subseteq I\\:|J| = k}} \left(\tfrac{\bias}{1-\bias}\right)^k\\
            &= \Pr[\RO=\ro] \cdot \sum_{k = 0}^{s} {s \choose k} \cdot \left(\tfrac{\bias}{1-\bias}\right)^k \cdot \left(\tfrac{1-p}{1-p}\right)^{s-k}\\
            &= \Pr[\RO=\ro] \cdot \left(\tfrac{p}{1-\bias} + \tfrac{1-p}{1-p}\right)^{s}\\
            &= \Pr[\RO=\ro] \cdot \left(\tfrac{1}{1-\bias}\right)^{s}\,.
        \end{align*}
        Plugging this back into the above inequality gives
        \begin{align*}
            1 &\geq \sum_{(\ro, \bfx) \in S} \Pr[T_{\ro, \bfx}]\\
            &\geq \sum_{(\ro, \bfx) \in S}  \Pr[\bfX = \bfx|\RO = \ro]\cdot \frac{1}{2} \cdot \sum_{(\ro', \bfx) \in T_{\ro, \bfx}} \Pr[\RO=\ro'] \\
            &= 
            \sum_{(\ro, \bfx) \in S}  \Pr[\bfX = \bfx|\RO = \ro] \cdot \frac{1}{2} \cdot \Pr[\RO=\ro]  \cdot \left(\tfrac{1}{1-\bias}\right)^{s}\\
            &= \Pr[S] \cdot \frac{1}{2} \cdot \left(\tfrac{1}{1-\bias}\right)^{s}\,.
        \end{align*}
    Rearranging this last display shows that the conditions of \Cref{thm:must-heavy-query} are satisfied with probability $\Pr[S] \leq 2\cdot (1-\bias)^{s}$.
\end{proof}

\subsection{Query-bounded adversaries cannot heavily query a correct answer.}

Here, we prove that if the adversary's query complexity is $2^{o(\secp^c)}$, then with overwhelming probability over $\RO$, $\cA^\RO$ does not \textit{heavily query} any correct answer $\bfx$. Heavy querying means that many symbols of $\bfx$ each receive large query norm from $\cA^\RO$.

\begin{lemma}\label[lemma]{thm:2-query-security}
    Let $\alpha = \alpha(\secp)$ be any function satisfying $\alpha(\secp) \ge 2^{-o(\secp^c)}$. Let $Q = Q(\secp)$ be any function satisfying $Q(\secp) = 2^{o(\secp^c)}$.  
    Let $\cA$ be any quantum  query algorithm making at most $Q$ queries.
    Then, with probability at least $1-L\cdot (1-\bias)^{4n/5}$ over the randomness of $\RO$, $\RO$ satisfies the following condition: 
    \begin{quote}
        For all $\bfx$ such that $\bfx$ is correct ($\bfx \in C$ and $\RO(\bfx) = \mathbf{0}$), there exists a set $I \subseteq [n]$, $|I|= \lceil{n/10\rceil}$, such that for all $i \in I$,
        \[\widetilde{w}_{i, x_i}(\RO) < \frac{\alpha}{4 n}\,.\]
    \end{quote}
\end{lemma}

\begin{proof}Put $s = \lceil{n/10\rceil}$.
    Let $S$ be the set of oracles that violate the condition of \Cref{thm:2-query-security}. Our goal is to show that $\Pr[S] \leq L\cdot (1-\bias)^{4n/5}$. This implies that with probability $\geq 1 - L\cdot (1-\bias)^{4n/5}$, over the randomness of $H$, the condition of \Cref{thm:2-query-security} is satisfied.

    \paragraph{Constructing the bad set:} For any $\ro \in S$, we will construct a set of oracles $T_{\ro}$ that includes $\ro$ as well as exponentially-many ``bad'' oracles $\ro'$ where $\cA^{\ro'}$ wastes query norm on an incorrect answer $\bfx_\ro$. 

    First, given an $\ro\in S$, pick a value $\bfx_\ro = (x_1, \ldots, x_n)$ and a set $I_\ro \subseteq [n]$ such that: 
    \begin{enumerate}
        \item $\bfx_\ro$ is correct ($\bfx_\ro \in C$ and $\ro(\bfx_\ro) = \mathbf{0}$), and
        \item $|I_\ro| \geq n - s + 1$, and for all $i \in I_\ro$,
        \[\widetilde{w}_{i, x_i}(\ro) \geq \frac{\alpha}{4n}\,.\]
    \end{enumerate}
    Such values of $(\bfx_\ro, I_\ro)$ must exist because $h$ violates the condition of \Cref{thm:2-query-security}. 

    Second, we apply \Cref{lemma:subcube with square root heavy coordinates} with $M = \{(i,x_i): i\in I_\ro\}$, $\beta = \tfrac{\alpha}{4n}$, $k := \lceil{\sqrt{n}/5\rceil}$, and sufficiently large $n$. 
   
    Then there exists a set $A_k \subseteq M$ such that 
    $$|M\setminus A_k|\ge n-s+1-\binom{k}{2} \ge \frac{4n}{5}$$
    and for each $\ro'\in \F_\ro(M\setminus A_k)$ there are at least $\sqrt{n}/5$ values of $i\in I_\ro$ for which 
    $$\widetilde{w}_{i,x_i}(h')\ge \frac{\beta}{(8Q|M|)^2} \ge  \frac{\beta}{(8Qn)^2} = \frac{\alpha}{256Q^2n^3}.$$
    
    Third, we take 
    $$T_{\ro} = \F_\ro(M\setminus A_k).$$

    Finally, let us say that \textbf{$\cA^\ro$ heavily queries $\bfx$} if there exists a set $I_1 \subseteq [n]$ of size $|I_1| \geq \sqrt{n}/5$ such that for all $i \in I_1$,
    \[ \widetilde{w}_{i, x_i}(\ro) \geq \frac{\alpha}{256Q^2n^3}\,.\]
    The previous discussion shows that for each $\ro' \in T_\ro$, $\cA^{\ro'}$ heavily queries $\bfx_\ro$.

    \begin{claim}
    For sufficiently large $\secp$, the number of codewords that $\cA^{\ro'}$ heavily queries is at most $L$.
    \end{claim}
    \begin{proofofclaim}
    By \Cref{lemma:number_of_heavily_queried_coordinates}, for any $\eta>0$, there are at most $(Q^2/\eta^2)$ inputs $(i,x)\in [n]\times \Sigma$ with query norm $\widetilde{w}_{i, x}(\ro') \geq \eta$. Picking $\eta = \frac{\alpha}{256Q^2n^3}$, gives, for sufficiently large $\secp$, at most \[
    Q^2 \cdot \left(\frac{256 Q^2 n^3}{\alpha}\right)^2 \le 2^{o(\lambda^c)} \le \ell\;
    \]
    coordinates $(i,x)$ such that $\widetilde{w}_{i,x}(\ro) \ge \frac{\alpha}{256 Q^2 n^3}$.
   
         If $\cA^{\ro'}$ heavily queries $\bfx$ then at least $\sqrt{n}/5$ coordinates of $\bfx$ belong to the above list of inputs that have query norm $\geq \frac{\alpha}{256 Q^2 n^3}$.
        Then, the list recoverability property of $C$ (\Cref{def:list-recovery,lem:suitablecodes}) says that the number of codewords $\bfx$ that satisfy this property is $\leq L$.
    \end{proofofclaim}

    \begin{claim}
        Any $\ro'$ belongs to $T_{\ro}$ for at most $L$ values of $\ro \in S$.
    \end{claim}
    \begin{proofofclaim}
        Assume toward contradiction that there are $L+1$ distinct values of $\ro \in S$ such that $\ro' \in T_{\ro}$. Call them $\{h_j: j \in [L+1]\}$. Let $\bfx_j$, also known as $\bfx_{\ro_j}$, be the string used to construct $T_{h_j}$. We know that $\cA^{\ro'}$ heavily queries each $\bfx_j$, but there are at most $L$ codewords that $\cA^{\ro'}$ heavily queries. By the pigeonhole principle, there are two different values $j, j' \in [L+1]$ such that $\bfx_j = \bfx_{j'}$.
        
        Furthermore, $h_{j}$ can be constructed from $\ro'$ and $\bfx_{j}$ by reprogramming $\ro'$ to map $\bfx_{j}$ to $\mathbf{0}$. By the same procedure, $h_{j'}$ can be constructed from $\ro'$ and $\bfx_{j'}$. Since $\bfx_j = \bfx_{j'}$, this implies that $h_j = h_{j'}$ because they are constructed from $(\ro', \bfx_j)$ by the same procedure. In summary $(h_j, \bfx_j) = (h_{j'}, \bfx_{j'})$---a contradiction. Therefore the initial assumption must be false, so in fact there are at most $L$ values of $\ro\in S$ for which $\ro' \in T_{\ro}$.
    \end{proofofclaim}

 Now we can finish the proof. Since each $\ro' \in \cH$ belongs to at most $L$ different sets $T_{\ro}$,
    \begin{align*}
        L &\geq \sum_{\ro \in S} \Pr[\RO\in T_{\ro}]\\
        &= \sum_{\ro\in S}\Pr[\RO=\ro] \cdot \sum_{\ro'\in T_{\ro}} \frac{\Pr[\RO=\ro']}{\Pr[\RO=\ro]}\,.
    \end{align*}
    We claim that for any $\ro\in S$,
    \[\sum_{\ro'\in T_{\ro}} \frac{\Pr[\RO=\ro']}{\Pr[\RO=\ro]} \ge \left(\tfrac{1}{1-\bias}\right)^{n-s+1-\binom{k}{2}}\,.
    \]
This holds since $T_{\ro}$ is obtained by considering the reprogrammings of all subsets of $M' = M\setminus A_k$. For each $J\subseteq M'$ we have 
$$\frac{\Pr[\RO=\ro\oplus J]}{\Pr[\RO=\ro]} = \left(\frac{\bias}{1-\bias}\right)^{|J|}.$$
Overall,
\[\sum_{\ro'\in T_{\ro}} \frac{\Pr[\RO=\ro']}{\Pr[\RO=\ro]}
= \sum_{J\subseteq M'}\frac{\Pr[\RO=\ro\oplus J]}{\Pr[\RO=\ro]}
= \sum_{J\subseteq M'}\left(\tfrac{\bias}{1-\bias}\right)^{|J|} = \left(\tfrac{1}{1-\bias}\right)^{|M'|}\,.\]
Plugging this into the above inequality gives
\begin{align*}L &\geq \sum_{\ro\in S} \Pr[\RO\in T_{\ro}]\\
        &= \sum_{\ro\in S}\Pr[\RO=\ro] \cdot \sum_{\ro'\in T_{\ro}} \frac{\Pr[\RO=\ro']}{\Pr[\RO=\ro]}
        \\
        &\ge \Pr[S]\cdot \left(\tfrac{1}{1-\bias}\right)^{|M'|}, 
        \end{align*}
        which after rearranging yields
        \[\Pr[S]\le L \cdot (1-\bias)^{|M'|} \le L \cdot (1-\bias)^{4n/5}.\qedhere\]
\end{proof}

\subsection{Finishing the proof}

Let us assume toward contradiction that \Cref{thm:CR-for-depth-bounded-adversaries} is false. 
Then there is some $Q = 2^{o(\secp^c)}$, some $\minent = o(\secp^{c})$, some adversary $\cA$ making $Q$ queries, and some inverse polynomial function $\delta(\secp)$ such that 

%
\begin{align*}
    \Pr_{\RO}\left[\left(\Pr\left[\Verify^\RO [1^\secp, 1^\minent, \cA^\RO(1^\secp, 1^\minent)] \neq \bot\right] \geq \delta(\secp)\right) \land \left(\Minent\left(\cA^\RO_\top (1^\secp, 1^\minent)\right) \leq \minent(\secp)\right)\right] = \beta(\secp)
\end{align*}
for some non-negligible $\beta(\secp)$.

Let us define
$\alpha(\secp) = \delta(\secp) \cdot 2^{-\minent(\secp)}$ and recall that $s = \lceil{n/10\rceil}$.

Let us also define two events based on $\RO$:
\begin{itemize}
    \item Let $\text{Event}_1$ be the event that $\RO$ satisfies:
    \[\Pr\left[\Verify^\RO [1^\secp, 1^\minent, \cA^\RO(1^\secp, 1^\minent)] \neq \bot\right] \geq \delta(\secp) \land \Minent\left(\cA^\RO_\top (1^\secp, 1^\minent)\right) \leq \minent(\secp).\]
    (Recall that we assumed this event happens with some non-negligible probability $\beta(\secp)$.)
    \item Let $\text{Event}_2$ be the event that $\RO$ satisfies the following condition:

    For all $\bfx$ such that $\bfx$ is correct ($\bfx \in C$ and $\RO(\bfx) = \mathbf{0}$), there exists a set $I \subseteq [n]$ of size $s$ such that for all $i \in I$,
        \[\widetilde{w}_{i, x_i}(\RO) < \frac{\alpha}{4 n}\,.\]
\end{itemize}

Finally, let us define an event based on $(\RO, \bfX)$ as the event  that the conditions of \Cref{thm:must-heavy-query} (repeated below for convenience) are satisfied.
\begin{itemize}
    \item Let $\text{Event}_3$ be the event that $(\RO, \bfX)$ have values $(\ro, \bfx)$ that satisfy:
    \begin{enumerate}
        \item $\cA^\ro$ gives a small total query norm to at least $s$ positions of $\bfx$. I.e., there exists a set $I \subseteq [n]$ of size $s$ such that for all $i \in I$,
        \[\widetilde{w}_{i, x_i}(\RO) < \frac{\alpha}{4 n}\,.\]
        \item $\bfx$ has high probability given $\ro$: $\Pr[\bfX = \bfx |\RO = \ro] \geq \alpha$.
        \item $\bfx$ is correct: $\bfx \in C$ and $\ro(\bfx) = \mathbf{0}$.
    \end{enumerate}
\end{itemize}

As $\alpha(\secp) = \delta(\lambda) \cdot 2^{-\minent(\secp)} = 2^{-o(\secp^c)}$ satisfies the conditions of \Cref{thm:2-query-security}
we get that     $\Pr_\RO[\text{Event}_2]$ is overwhelming, namely \[\Pr_\RO[\text{Event}_2] \ge 1-L\cdot (1-\bias)^{4n/5}\,.\]



With that, we derive a lower bound on $\Pr_{\RO}[\text{Event}_1 \land \text{Event}_2]$.
    \begin{align*}
        \Pr_\RO[\text{Event}_1 \land \text{Event}_2] &\geq \Pr_\RO[\text{Event}_1] - \Pr_\RO[\neg \text{Event}_2]\\
        &\geq \beta(\secp) - L\cdot (1-\bias)^{4n/5}\\ &\geq \beta(\secp)/2\,.
    \end{align*}
    
    where in the last inequality we used the fact that $\beta(\secp)$ is non-negligible while $ L\cdot (1-\bias)^{4n/5}$ is negligible.

Toward contradiction, our goal will be to derive a lower bound on $\Pr_{\RO, \bfX}[\text{Event}_3]$ showing it is at least $\Pr_\RO[\text{Event}_1 \land \text{Event}_2]\cdot \alpha(\secp)$.
Toward this end, we prove the following lemma.
\begin{lemma}\label[lemma]{lemma:E1 and E2 implies E3 for some xh}
    Let $\ro$ be any value of $\RO$ that satisfies $\text{Event}_1$ and $\text{Event}_2$. Then there exists an $\bfx^\ro$ such that $(\ro, \bfx^\ro)$ satisfy $\text{Event}_3$.
\end{lemma}
\begin{proof}
The first condition of $\text{Event}_1$ is $\Pr\left[\Verify^\RO [1^\secp, 1^\minent, \cA^\RO(1^\secp, 1^\minent)] \neq \bot\right] \geq \delta(\secp)$, which is equivalent to:
    \[\Pr_\bfX[\bfX \in C \land \RO(\bfX) = \mathbf{0} | \RO = \ro] \geq \delta(\secp).\]
Next, 
\begin{align*}
    \Minent\left(\cA^\RO_\top (1^\secp, 1^\minent)\right) &= - \log\left(\max_{\bfx} \Pr_\bfX\left[\bfX = \bfx | \RO = \ro, (\bfX \in C \land \RO(\bfX) = \mathbf{0})\right]\right)\,.
\end{align*}
Furthermore, for any $\bfx$ that is correct with respect to $\ro$ ($\bfx \in C \land \ro(\bfx) = \mathbf{0}$),
\begin{align*}
    \Pr_\bfX\left[\bfX = \bfx | \RO = \ro, (\bfX \in C \land \RO(\bfX) = \mathbf{0})\right] &= \frac{\Pr_\bfX\left[\bfX = \bfx, (\bfX \in C \land \RO(\bfX) = \mathbf{0}) | \RO = \ro\right]}{\Pr_\bfX\left[\bfX \in C \land \RO(\bfX) = \mathbf{0}|\RO = \ro\right]}\\
    &= \frac{\Pr_\bfX\left[\bfX = \bfx | \RO = \ro\right]}{\Pr_\bfX\left[\bfX \in C \land \RO(\bfX) = \mathbf{0}|\RO = \ro\right]}\\
    &\leq \frac{\Pr_\bfX\left[\bfX = \bfx | \RO = \ro\right]}{\delta(\secp)}\,.
\end{align*}

The second condition of $\text{Event}_1$ is $\Minent\left(\cA^\RO_\top (1^\secp, 1^\minent)\right) \leq \minent(\secp)$, which implies that:
\begin{align*}
    -\log\left(\max_{\bfx} \Pr_\bfX[\bfX = \bfx | \RO = \ro, (\bfX \in C \land \RO(\bfX) = \mathbf{0})]\right) &\leq \minent(\secp),\\
    \text{or equivalently,}\qquad \max_{\bfx} \Pr_{\bfX}[\bfX = \bfx | \RO = \ro, (\bfX \in C \land \RO(\bfX) = \mathbf{0})] &\geq 2^{-\minent(\secp)}\,.
\end{align*}
Then there exists an $\bfx^\ro$ such that
\begin{align*}
    \Pr_\bfX[\bfX = \bfx^\ro | \RO = \ro, (\bfX \in C \land \RO(\bfX) = \mathbf{0})] &\geq 2^{-\minent(\secp)} > 0\,.
\end{align*}
Furthermore, this $\bfx^\ro$ is correct ($\bfx^\ro \in C \land \ro(\bfx^\ro) = \mathbf{0}$) because otherwise, $\Pr_\bfX[\bfX = \bfx^\ro | \RO = \ro, (\bfX \in C \land \RO(\bfX) = \mathbf{0})] = 0$. Then
\begin{align*}
    \Pr_\bfX\left[\bfX = \bfx^\ro | \RO = \ro\right] &\geq \delta(\lambda) \cdot \Pr_\bfX[\bfX = \bfx^\ro | \RO = \ro, (\bfX \in C \land \RO(\bfX) = \mathbf{0})]\\
    &\geq \delta(\lambda) \cdot 2^{-\minent(\secp)}
    \\
    &= \alpha(\secp)\,.
\end{align*}

In summary, if $\text{Event}_1$ occurs, then there exists an $\bfx^\ro$ such that:
\begin{itemize}
    \item $\bfx^\ro$ is correct: $\bfx^\ro \in C \land \ro(\bfx^\ro) = \mathbf{0}$, and 
    \item $\Pr_\bfX\left[\bfX = \bfx^\ro | \RO = \ro\right] \geq \alpha(\secp)$.
\end{itemize}

Furthermore, since this $\bfx^\ro$ is correct, $\text{Event}_2$ implies that this $\bfx^\ro$-value also satisfies the following condition:
\begin{quote}
    There exists a set $I \subseteq [n]$ of size $s$ such that for all $i \in I$,
        \[\widetilde{w}_{i, \bfx^{\ro}_i}(\RO) < \frac{\alpha}{4 n}\,.\]
\end{quote}
This shows that $(h, \bfx^h)$ satisfy all the conditions of $\text{Event}_3$.
\end{proof}

    By \Cref{lemma:E1 and E2 implies E3 for some xh}, if $\text{Event}_1 \land \text{Event}_2$ occur for some $\ro$, then there exists some special $\bfx^\ro$ such that $(\ro, \bfx^\ro)$ satisfy the conditions of $\text{Event}_3$. Furthermore, $\text{Event}_3$ guarantees that:
    \[\Pr_\bfX\left[\bfX = \bfx^\ro | \RO = \ro\right] \geq \alpha(\secp)\,.\]
    
    One way to satisfy the conditions of $\text{Event}_3$ is to sample an $\ro$ that satisfies the conditions of $\text{Event}_1 \land \text{Event}_2$ and then to sample the particular $\bfx^\ro$ value that is guaranteed to satisfy $\text{Event}_3$. The probability of sampling values $(\bfh, \bfx^\ro)$ of this form is
    \[\geq \Pr_\RO[\text{Event}_1 \land \text{Event}_2] \cdot \alpha(\secp) \geq \frac{\beta(\secp)\cdot \alpha(\secp)}{2} \]
    and thus 
    \[\Pr_{\RO, \bfX}[\text{Event}_3] \geq \frac{\beta(\secp)\cdot \alpha(\secp)}{2}.\]
    
On the other hand, by \Cref{thm:must-heavy-query}, 
\[\Pr_{\RO, \bfX}[\text{Event}_3] \le 2\cdot (1-\bias)^{s}\,.\]
Combining the two inequalities, we get that 
\[\beta(\secp) \le \tfrac{4}{\alpha(\secp)} \cdot (1-\bias)^{s} \le O(e^{-\bias n/10}\cdot 2^{o(\secp^c)}) = \exp(-\Omega(\sqrt n)) = \negl(\secp),\]
which is a contradiction to the fact that $\beta(\secp)$ is non-negligible.

\section{AI Disclosure}

ChatGPT 5.6 was used to refine the parameter choices of the family of FRS codes. All other techniques in
this work were human-generated. The authors also brainstormed ideas in conversation with AI, but none of those techniques were used in the paper.
Brainstorming with AI brought a lot of interesting approaches and a counterexample to stronger variants of \Cref{q:heavy-reprogramming}, but so far the question remains open.

\section{Acknowledgments}\label[section]{sec:acknowledgements}
A.C. thanks Boyang Chen for helpful discussions and is grateful for support from the Google Research Scholar program.
D.K. was supported in part by AFOSR, NSF
CNS-2247727, and a Google Research Scholar award.
S.M. thanks Henry Yuen for valuable discussions.
A.T. was supported by NSF CAREER award CCF-2145474.
This material is based upon work
supported by the Air Force Office of Scientific Research under award number FA9550-23-1-0543.

\bibliographystyle{alpha}
\bibliography{references.bib}

\appendix
\section{Proof of Corollary~\ref{cor:CR-for-depth-bounded-adversaries-qrom}}\label[section]{sec:app}
Here, we show that the min-entropy guarantee of~\Cref{thm:CR-for-depth-bounded-adversaries} also holds, with identical asymptotic parameters, relative to a standard random oracle that is used to instantiate the biased random oracle of~\Cref{construction:CR-bounded-min-entropy}.

The instantiation from a standard random oracle $F$ can be thought of as an encoding
\[
\mathsf{Enc}:\mathcal F\longrightarrow \mathcal H,
\qquad
H_F:=\mathsf{Enc}(F),
\]
where \(F\) is uniform over regular random oracles and \(H_F\) has the desired \(p\)-biased random-oracle distribution. This encoding is the one described at the start of \Cref{sec:QROM}: we assume $p = 2^{-d}$ for some integer $d$; then $H_F(x)$ is defined by representing $x$ as a bitstring, truncating $F(x)$ to $d$ bits, and finally setting $H_F(x) = 1$ if and only if $F(x) = 1^d$. This way, $\Pr[H(x) = 1] = \Pr[F(x) = 1^d] = 2^{-d} = p$.

The encoding map is not invertible, but we can define a map that goes the other way, by including precisely the extra randomness in $F$ that is not determined by $H_F$:
\[
\mathsf{Dec}:\mathcal H\times\mathcal R\longrightarrow \mathcal F,
\qquad
F_{h,r}:=\mathsf{Dec}(h,r),
\]
with the following properties. First,
\[
\mathsf{Enc}(F_{h,r})=h
\qquad
\text{for every }h,r.
\]
Second, for \(H\) drawn from the biased distribution and \(R\) drawn independently,
\begin{equation}
\label{eq:distr-identity}
(H,F_{H,R})
\stackrel{d}{=}
(\mathsf{Enc}(F),F),
\qquad F\leftarrow\text{uniform}.
\end{equation}
Third, for every \emph{fixed} \(r\), a query to \(F_{h,r}\) can be simulated with \emph{two} queries to \(h\): compute $h(x)$, use it to answer the $F_{h,r}$ query, then uncompute $h$.

\medskip
Now, our goal is to establish a min-entropy guarantee on any quantum algorithm $B$ that makes queries to $F$ and solves the problem of \Cref{construction:CR-bounded-min-entropy}, where $F$ is used to instantiate the biased oracle $H$. 

Let us fix any functions $Q(\secp) = 2^{o(\secp^c)}$ and $h_\infty(\secp) = o(\secp^c)$, any $F$-oracle algorithm $B$ that makes at most $Q$ queries to $F$, and any inverse polynomial function $\delta$. Then for any \(r\), define an \(H\)-oracle algorithm $A_r$ as follows:
\begin{equation}
\label{eq:3}
A_r^h:=B^{F_{h,r}} \,.
\end{equation}
$A_r$ simulates each query of \(B\) to $F_{h,r}$ by making two queries to $h$, as described above. Note that $A_r$ makes at most $Q' := 2Q = 2^{o(\secp^c)}$ queries to $H$. 

Additionally, the output of $B^{F_{h,r}}$ has an identical distribution to the output of $A_r^h$. In particular, letting $B_{\top}^{F_{h,r}}$ and $A_{\top,r}^h$ denote the respective distributions of outputs \emph{conditioned} on the output being correct, we have
\[
\max_z \Pr[B_{\top}^{F_{h,r}}=z]= \max_z\Pr[A_{\top,r}^h=z] \,.
\]
Similarly, if validity of an answer is described by the relation \(V(h,z)\), then
\[
\Pr[V(\mathsf{Enc}(F_{h,r}),B^{F_{h,r}}) \neq \bot] = \Pr[V(h,A_r^h) \neq \bot]\,.
\]

\medskip
Then the proof of \Cref{thm:CR-for-depth-bounded-adversaries} implies that there exists a negligible function $\beta(\secp)$ such that for any $r$ and $\secp$,
\begin{equation}
\Pr_{H}\left[\Pr[V(H,A_r^H)\neq \bot] \geq \delta(\secp) \,\,\land \,\,
\max_z\Pr[A_{\top,r}^H=z]\geq2^{-h_\infty(\secp)}
\right] \leq \beta(\secp) \,, \label{eq:22}
\end{equation}
Note that we can use the same negligible function $\beta$ for any adversary $A_r$ because the proof of \Cref{thm:CR-for-depth-bounded-adversaries} derives the same upper-bound $\beta$ for every $Q'$-query adversary.

Then, averaging over $r \gets R$, and using the distribution identity \eqref{eq:distr-identity}, we get
\[
\begin{aligned}
\Pr_F&\left[\Pr[V(\mathsf{Enc}(F),B^F) \neq \bot] \geq \delta(\secp) \,\, \land \,\, \max_z \Pr[B_{\top}^F=z]\geq2^{-h_\infty(\secp)}\right] \\
&=
\Pr_{H,R}\left[\Pr[V(\mathsf{Enc}(F_{H,R}),B^{F_{H,R}}) \neq \bot] \geq \delta(\secp) \,\, \land \,\, \max_z \Pr[B_{\top}^{F_{H,R}}=z]\geq2^{-h_\infty(\secp)}\right]
\\
&=
\mathbb E_R
\Pr_H\left[\Pr[V(H,A_R^H) \neq \bot] \geq \delta(\secp) \,\, \land \,\, \max_z\Pr[A_{\top, R}^H=z]\geq2^{-h_\infty(\secp)}\right]
\\
&\leq \bbE_R[\beta(\secp)] = \beta(\secp).
\end{aligned}
\]
where the last inequality holds by the guarantee of Equation~\eqref{eq:22}.

This shows that the min-entropy guarantee of Equation~\eqref{eq:22} transfers with the same parameters to any $F$-oracle algorithm $B$, up to a factor of $\frac12$ in the number of queries, which does not affect the asymptotic statement.

\end{document}